%% file: paper.tex
\pdfoutput=1
\documentclass[runningheads]{llncs} 
\input{preamble}

\titlerunning{Behavioral Program Logic}

\begin{document}
\maketitle
\begin{abstract}
We present Behavioral Program Logic (BPL), a dynamic logic for trace properties that incorporates concepts from behavioral types and allows reasoning about non-functional properties within a sequent calculus.
BPL uses \emph{behavioral modalities} $[\statement \halfsim \tau]$, to verify statements $\statement$ against \emph{behavioral specifications} $\tau$.
Behavioral specifications generalize postconditions and behavioral types. 
They can be used to specify other static analyses, e.g., data flow analyses.
This enables deductive reasoning about the results of multiple analyses on the same program, potentially implemented in different formalisms.
Our calculus for BPL verifies the behavioral specification gradually, as common for behavioral types.
This vastly simplifies specification, calculus and composition of local results.
We present a sequent calculus for object-oriented actors with futures that integrates a pointer analysis and bridges the gap between behavioral types and deductive verification.

This technical report introduces (1) complete LAGC semantics of a Core Active Object language (CAO) without continuations (2) Behavioral Program Logic and (3) gives an example for a behavioral type expressed in Behavioral Program Logic, \emph{method types}.
This report contains the soundness proofs for method types.
While the semantics cover CAO with suspension, the method types do not, to simplify the presentation.
\end{abstract}
\bibliographystyle{acm}
\section{Introduction}
When reasoning about concurrent programs, the intermediate states of an execution are of more relevance than when reasoning about sequential programs.
In an object-oriented setting, it does not suffice to specify pre- and postcondition of some method \methodname. 
Instead, the \emph{traces} generated by \methodname must be specified. 

Recently, dynamic logics for trace properties have been developed~\cite{bern,tracedl,absdl2} to leverage well-established verification techniques from dynamic logic~\cite{keybook} to a concurrent setting.
The application of these approaches to real world models of distributed systems~\cite{DinTHJ15,KamburjanH17} revealed two shortcomings:
(1) the composition of method-local verification results to a guarantee of the whole system is not automatic and 
(2) the specification of trace properties is too complex.
Thus, the current approaches are deemed as not practical for serious verification efforts.

Another group of verification techniques, \emph{behavioral types}, aim ``{\it to describe properties associated with the behavior of programs and in this way also describe how a computation proceeds.}''~\cite{hut}.  
For object-oriented languages, behavioral types can also be seen as specifications of traces of methods.
Behavioral types, especially session types~\cite{Honda08}, are restricted in their expressive power to easily compose their local results to global guarantees, and are natural specifications for protocols. 
However, they lack precision when handling state~\cite{Bocchi12} or require additional static analyses~\cite{ifm}.
For Active Objects~\cite{boer} (object-oriented actors with futures), a translation from session types to a trace logic has been given~\cite{ifm}.

We introduce \emph{Behavioral Program Logic} (\TPL) to combine precise state reasoning from program logics with the relative simplicity of behavioral types and enable the integration of static analyses into deductive reasoning.
The main difference to previous approaches in dynamic logic for trace properties is the 
\emph{behavioral modality} $[\statement \halfsim \tau]$, which expresses that all traces of statement \statement satisfy specification $\tau$. 
The specification $\tau$ is not a formula, as the postcondition of modalities in classical dynamic logic, but is a specification translated into a monadic second order formula over traces. 
Similarly to behavioral types, $\tau$ may contain syntactic elements and allows to syntactically match with $\statement$. Sequent calculi for \TPL may reduce $\statement$ \emph{and} $\tau$ in one rule.
Contrary to previous dynamic logics for traces, behavioral specifications are more succinct and easier to compose and decompose by, e.g., using the projection mechanism of session types.

We distinguish between behavioral types, that have a sequent calculus of the above kind, and \emph{behavioral specifications}, which do not.
Behavioral specifications interface with external properties, such as a data-flow points-to analysis. 
Beyond integrating external analyses into the sequent calculus, this modularizes the sequent calculus by expressing different properties with different behavioral specifications.
Behavioral specifications are clear interfaces that allow to close proofs once more context is known and generalize proof repositories~\cite{repos}.

Our main contributions are (1) \TPL, a trace program logic that integrates deductive reasoning with static analyses
(2) \emph{method types}, a behavioral type in \TPL that generalizes method contracts, object invariants and local types for Active Objects.
Decomposition and projection for Session Types for Active Objects are given in~\cite{ifm,icfem}.
We introduce our programming language in Sec.~\ref{ch:ao} and \TPL in Sec.~\ref{sec:bpl1}--\ref{sec:idea}.
In Sec.~\ref{sec:symbol} we introduce method types. 
Sec.~\ref{sec:other} summarizes previous approaches
and concludes. 

\input{chapters/03_active}
\input{chapters/04_bpl}

\section{Conclusion and Related Work}\label{sec:other}
This works presents \TPL, a program logic for object-oriented distributed programs 
that enables deductive reasoning about the results of static analyses and integrates concepts from behavioral types by pattern-matching statement and specification. 
The \emph{method type} behavioral type generalizes method contracts, session types and object invariants.
In the following, we discuss related work. 

\paragraph{Dynamic Logics.}
Beckert and Bruns~\cite{bern} use LTL formulas in dynamic logic modalities in their Dynamic Trace Logic (DTL) for Java. Given an LTL formula $\phi$, the DTL-formula $[\statement]\phi$ expresses that $\phi$ describes all traces of $\statement$. 
DTL uses a restricted form of pattern matching: its three loop invariant rules depend on the outermost operator of $\phi$ and other rules may consume a\`{ }next'' operator. 
DTL does not use events and specifies patterns of state changes, not of interactions.

The Abstract Behavior Specification Dynamic Logic (ABSDL) of Din and Owe~\cite{absdl2} is for the ABS language~\cite{abs}.
In ABSDL, a formula $[\statement]\phi$, where $\phi$ is a first-order formula over the program state, has the standard meaning that $\phi$ holds after $\statement$ is executed.
ABSDL uses a special program variable to keep track of the visible events. 
Its rules are tightly coupled with object-invariant reasoning. This makes it impossible to specify the state at arbitrary interactions. 

Bubel et al.~\cite{tracedl} define dynamic logic with coinductive traces (DLCT).
In DLCT,  a formula $[\statement]\phi$, where $\phi$ is a trace modality formula, containing symbolic trace formulas, has the meaning that every trace of $\statement$ is a model for $\phi$.
Contrary to ABSDL, DLCT keeps track of the whole trace, not just the events.
DLCT is not able to specify the property that between two states, some form of event does \emph{not} occur, as symbolic trace formulas are not closed under negation.

\paragraph{Behavioral Types.}
A number of behavioral types deals with assertions~\cite{berger1,Bocchi12} or Actors~\cite{elena,HenrioLM17,NeykovaY16}. 
Stateful Behavioral Types for Active Objects (STAO)~\cite{ifm} uses both and defines the judgment $\phi, \statement' \vdash \statement : \tau$, 
that expresses that all traces of $\statement$ are models for the translation of $\tau$. $\phi$ and $\statement'$ keep track of the chosen path so far.
STAO is not able to reason about multiple judgments, but relies on external analyses for precision. Reasoning about these results happens on a meta-level.

Finally, Propositions-as-Types theorems (PaT) have been established~\cite{intu1,intu2} between  session types for the $\pi$-calculus and intuitionistic linear logic.
They are specific to this setting and do not characterize general behavioral types. To our best knowledge, Def.~\ref{def:behtype} is the first formal characterization of behavioral types.

\paragraph{Future Work.}
An implementation of \TPL for full ABS is ongoing and as future work, we plan to investigate further types and concurrency models, in particular systems with shared memory and effect type systems.

\bibliography{paper.bib}

\end{document}

%% file: preamble.tex
\usepackage{amsmath}

\usepackage{amssymb,amsthm, amsfonts}
\usepackage{bussproofs}
\usepackage{stmaryrd}
\usepackage{epigraph}

\usepackage{dsfont}

\usepackage{bbm}

\usepackage{graphicx}
\usepackage{listings}
\usepackage{bussproofs}
\usepackage{tikz}
\usepackage{floatflt} 

\usepackage{ifthen}
\usepackage{xspace}
\usepackage{pdfpages}
\usepackage{todonotes}

\usepackage{pgf-umlsd}
\usetikzlibrary{arrows,shapes,trees,positioning,decorations,shapes}
\usetikzlibrary{decorations.markings,automata}

\input{abs}

\input{definitions}

\title{Behavioral Program Logic and LAGC Semantics without Continuations (Technical Report)} 
\author{Eduard Kamburjan}
\institute{Technische Universit{\"a}t Darmstadt, Darmstadt, Germany\\ Department of Computer Science, Software Engineering Group\\\texttt{kamburjan@cs.tu-darmstadt.de}}

%% file: abs.tex
\DeclareOldFontCommand{\rm}{\normalfont\rmfamily}{\mathrm}
\DeclareOldFontCommand{\sf}{\normalfont\sffamily}{\mathsf}
\DeclareOldFontCommand{\tt}{\normalfont\ttfamily}{\mathtt}
\DeclareOldFontCommand{\bf}{\normalfont\bfseries}{\mathbf}
\DeclareOldFontCommand{\it}{\normalfont\itshape}{\mathit}
\DeclareOldFontCommand{\sl}{\normalfont\slshape}{\@nomath\sl}
\DeclareOldFontCommand{\sc}{\normalfont\scshape}{\@nomath\sc}
\DeclareRobustCommand*\cal{\@fontswitch\relax\mathcal}
\DeclareRobustCommand*\mit{\@fontswitch\relax\mathnormal}

\colorlet{keywordcolor}{blue!50!black}
\colorlet{commentcolor}{green!60!black}
\colorlet{typecolor}{violet}
\newcommand{\sourcefont}{\ttfamily\small}
\newcommand{\commentfont}{\slshape\rmfamily\color{commentcolor}}

\lstdefinelanguage{ABS}{
        keywords={od,do,assert,this,new,data,type,def,case,of,local,class,interface,
        extends,implements,if,then,else,await,get,return,skip,while,module,
        import,export,from,to,suspend,delta,adds,modifies,removes,original,productline,
        features,core,corefeatures,optionalfeatures,after,when,product,hasAttribute,
        hasMethod,hasField,hasInterface,uses,root,extension,group,allof,oneof,require,
        stateupdate,object,main,objectupdate,classupdate,fi,
        exclude,original,ifin,ifout,opt,null,
        newgroup,data,thiscomp,in,joins,leaves,subtypeOf,acquire,except,as,component,Pre,Abs
        },
        keywordstyle=\color{keywordcolor}\bfseries\sffamily,
        morekeywords=[2]{Unit, Fut,Int, Bool, Rat, List, Set, Pair, Fut, Maybe, String, Triple, Either, Map},
        keywordstyle=[2]\color{typecolor},
        sensitive=true,
        comment=[l]{//},
        morecomment=[s]{/*}{*/},
        morestring=[b]"
}

\lstdefinelanguage[v9]{Java}[]{Java}{
        morekeywords={module,requires,provides,uses,with,to,exports}
}
\lstdefinelanguage[ContextJ]{Java}[]{Java}{
        morekeywords={layer,with,without,proceed,before,after}
}
\lstdefinelanguage[FOP]{Java}[]{Java}{
        morekeywords={refines,original,Super}
}
\lstdefinelanguage[JastAdd]{Java}[]{Java}{
        morekeywords={aspect,syn,inh,lazy}
}

\lstdefinestyle{code}{
        basicstyle=\sourcefont\upshape,
        keywordstyle=\color{keywordcolor}\bfseries\sffamily,
        commentstyle=\commentfont,
        columns=fullflexible,
        mathescape=true,
        escapechar={\#},
        keepspaces=true,
        showstringspaces=false,
        aboveskip=8pt, 
        numbers=left,
        stepnumber=1, 
        numberstyle=\ttfamily\scriptsize\color{gray},
        numbersep=4pt,
        xleftmargin=1.5em,
        xrightmargin=1.5em,
        framexleftmargin=1.2em,
        framexrightmargin=1em,
        framextopmargin=0.5ex,
        breaklines=true,
        breakindent=3pt,
}

\lstdefinestyle{abs}{
        style=code,
        language=ABS,
}
\lstdefinestyle{java}{
        style=code,
            language=Java
}
\lstdefinestyle{java9}{
        style=code,
            language=[v9]Java
}
\lstdefinestyle{aspectj}{
        style=code,
        language=[AspectJ]Java
}
\lstdefinestyle{jastadd}{
        style=code,
        language=[JastAdd]Java
}
\lstdefinestyle{contextj}{
        style=code,
        language=[ContextJ]Java
}
\lstdefinestyle{FOP}{
        style=code,
        language=[FOP]Java
}
\lstdefinestyle{scala}{
        style=code,
        language=Scala,
        morekeywords={self}
}

\newcommand{\code}[2][]{\lstinline[style=code,basicstyle=\ttfamily\upshape,#1]{#2}}

\newcommand{\abs}[2][]{\code[style=abs,#1]{#2}}

\lstnewenvironment{srccode}[1][]{
        \minipage{\linewidth}
        \lstset{style=code,
        backgroundcolor=\color{white},
        rulecolor=\color{gray!50},
        frame=tblr,
        captionpos=b,
        #1}
}{\endminipage}

\lstnewenvironment{abscode}[1][]{
        \minipage{1\linewidth}
        \lstset{style=abs,
        backgroundcolor=\color{white},
        rulecolor=\color{gray!50},
        frame=tblr,
        captionpos=b,
        #1}
}{\endminipage}

\lstnewenvironment{javacode}[1][]{
        \minipage{\linewidth}
        \lstset{style=java,
        backgroundcolor=\color{gray!5},
        rulecolor=\color{gray!50},
        frame=tblr,
        captionpos=b,
        #1}
}{\endminipage}

\lstnewenvironment{jastaddcode}[1][]{
        \minipage{\linewidth}
        \lstset{style=jastadd,
        backgroundcolor=\color{gray!5},
        rulecolor=\color{gray!50},
        frame=tblr,
        captionpos=b,
        #1}
}{\endminipage}

%% file: definitions.tex
\newcommand{\ektodo}[1]{\todo[color=brown!40, author=EK]{#1}}

\newcommand{\COMMENT}[1]{}

\makeatletter
\newsavebox{\@brx}
\newcommand{\llangle}[1][]{\savebox{\@brx}{\(\m@th{#1\langle}\)}%
  \mathopen{\copy\@brx\kern-0.5\wd\@brx\usebox{\@brx}}}
  \newcommand{\rrangle}[1][]{\savebox{\@brx}{\(\m@th{#1\rangle}\)}%
    \mathclose{\copy\@brx\kern-0.5\wd\@brx\usebox{\@brx}}}
    \makeatother

\let\temp\phi
\let\phi\varphi
\let\varphi\temp

\makeatletter
\newcommand{\xRightarrow}[2][]{\ext@arrow 0359\Rightarrowfill@{#1}{#2}}

\newcommand{\branchs}[1]{\ensuremath{
\left\{\begin{array}{l}
#1
\end{array}\right\}
}\xspace}

\newcommand{\cased}[1]{\ensuremath{
\left\{\begin{array}{ll}
#1
\end{array}\right.
}\xspace}

\newcommand{\casedr}[1]{\ensuremath{
\left.\begin{array}{rl}
#1
\end{array}\right\}\\
}\xspace}

\DeclareMathOperator*{\halfsim}{\Vdash}
\DeclareMathOperator*{\chop}{\ast\ast}
\DeclareMathOperator*{\megachop}{\ast\!\ast\!\ast}

\newcommand{\xabs}[1]{\text{\abs{#1}}}
\newcommand{\osynt}[1]{\mathsf{#1}}

\newcommand{\sep}{  \ | \ }
\newcommand{\many}[1]{\overrightarrow{#1}}

\newcommand{\coreactor}{\ensuremath{CAO}\xspace}

\newcommand{\sequence}[1]{\ensuremath{ \left\langle #1 \right\rangle \xspace}}

\newcommand{\prgm}{\ensuremath{\osynt{Prgm}}\xspace}
\newcommand{\classlang}{\ensuremath{\osynt{Class}}\xspace}
\newcommand{\mainlang}{\ensuremath{\osynt{Main}}\xspace}
\newcommand{\methodlang}{\ensuremath{\osynt{Meth}}\xspace}
\newcommand{\fieldlang}{\ensuremath{\osynt{Field}}\xspace}
\newcommand{\explang}{\ensuremath{\synt{e}}\xspace}
\newcommand{\sexplang}{\ensuremath{\symbolic{\explang}}\xspace}
\newcommand{\expr}{\ensuremath{\xabs{e}}\xspace}
\newcommand{\sexpr}{\ensuremath{\symbolic{\xabs{e}}}\xspace}

\newcommand{\objname}{\ensuremath{\synt{X}}\xspace}
\newcommand{\classname}{\ensuremath{\synt{C}}\xspace}
\newcommand{\methodname}{\ensuremath{\synt{m}}\xspace}

\newcommand{\method}{\mathsf{M}}

\newcommand{\varname}{\ensuremath{\xabs{v}}\xspace}

\newcommand{\mainABS}{\text{\abs{main}}}

\newcommand{\event}{\ensuremath{\mathsf{ev}}\xspace}
\newcommand{\invocev}{\ensuremath{\mathsf{invEv}}\xspace}
\newcommand{\invocrev}{\ensuremath{\mathsf{invREv}}\xspace}
\newcommand{\resolvev}{\ensuremath{\mathsf{futEv}}\xspace}
\newcommand{\resolvrev}{\ensuremath{\mathsf{futREv}}\xspace}
\newcommand{\noev}{\ensuremath{\mathsf{noEv}}\xspace}
\newcommand{\awaitev}{\ensuremath{\mathsf{condEv}}\xspace}
\newcommand{\awaitrev}{\ensuremath{\mathsf{condREv}}\xspace}
\newcommand{\suspev}{\ensuremath{\mathsf{suspEv}}\xspace}
\newcommand{\susprev}{\ensuremath{\mathsf{suspREv}}\xspace}

\newcommand{\future}{\ensuremath{f}}

\COMMENT{
\newcommand{\TINFER}[3]{\begin{array}{c}\rulename{#1}  \frac{\begin{array}{c}#2
\\[0.5mm]
\end{array}
}{
\begin{array}{c}
\\[-3.5mm]
\displaystyle{#3}
\end{array}
}\end{array}}

}

\newcommand{\rulename}[1]{\textbf{\scriptsize(\textsf{#1})}}
\newcommand{\TINFER}[4][]{
\begin{prooftree}
\AxiomC{\ensuremath{#3}}
\LeftLabel{\rulename{#2}}
\RightLabel{#1}
\UnaryInfC{\ensuremath{#4}}
\end{prooftree}
}
\def\fCenter{\ \Rightarrow\ }
\newcommand{\SINFER}[4][]{
\begin{prooftree}
\AxiomC{\ensuremath{#3}}
\LeftLabel{\rulename{#2}}
\RightLabel{#1}
\UnaryInfC{\ensuremath{#4}}
\end{prooftree}
}

\newcommand{\SSINFER}[5][]{
\begin{prooftree}
\AxiomC{\ensuremath{#3}}
\noLine
\UnaryInfC{\ensuremath{#4}}
\LeftLabel{\rulename{#2}}
\RightLabel{#1}
\UnaryInfC{\ensuremath{#5}}
\end{prooftree}
}

\newcommand{\synt}[1]{\ensuremath{\xabs{#1}}}
\newcommand{\dType}{\ensuremath{\xabs{D}}}
\newcommand{\statement}{\ensuremath{\xabs{s}}\xspace}

\newcommand{\type}{\ensuremath{\mathbbm{T}}\xspace}
\newcommand{\typesyn}{\ensuremath{\tau}\xspace}
\newcommand{\typecalc}{\ensuremath{\gamma}\xspace}
\newcommand{\typesem}{\ensuremath{\alpha}\xspace}
\newcommand{\typescheme}{\ensuremath{\iota}\xspace}
\newcommand{\state}{\ensuremath{\sigma}\xspace}

\newcommand{\heap}{\ensuremath{\rho}\xspace}
\newcommand{\heapid}{\ensuremath{\rho_\mathsf{id}}\xspace}

\newcommand{\trace}{\ensuremath{\theta}\xspace}
\newcommand{\globaltrace}{\ensuremath{\gamma}\xspace}
\newcommand{\traceset}{\ensuremath{\Theta}\xspace}

\newcommand{\history}{\ensuremath{\mathsf{hs}}\xspace}
\newcommand{\selection}{\ensuremath{\mathsf{sc}}\xspace}
\newcommand{\concr}{\ensuremath{\chi}\xspace}

\newcommand{\pttype}{\ensuremath{\type_\mathsf{p2}}\xspace}

\newcommand{\pttypesem}{\ensuremath{{\mathsf{p2}}}\xspace}
\newcommand{\ptsim}{\ensuremath{\halfsim\limits^{\pttypesem}}\xspace}

\newcommand{\posttype}{\ensuremath{\type_\mathsf{pst}}\xspace}

\newcommand{\posttypesem}{\ensuremath{{\mathsf{pst}}}\xspace}

\newcommand{\postsim}{\ensuremath{\halfsim\limits^{\posttypesem}}\xspace}

\newcommand{\localtype}{\ensuremath{\type_\mathsf{met}}\xspace}

\newcommand{\methtype}{\ensuremath{\mathsf{L}}\xspace}
\newcommand{\localtypesem}{\ensuremath{\typesem_\mathsf{met}}\xspace}

\newcommand{\localtypescheme}{\ensuremath{\typescheme_\mathsf{met}}\xspace}
\newcommand{\localsim}{\ensuremath{\halfsim\limits^{\localtypesem}}\xspace}
\newcommand{\localactchoice}[1]{\oplus\branchs{#1}\xspace}
\newcommand{\localpaschoice}[2]{\&(#1)\branchs{#2}\xspace}

\newcommand{\TPL}{\ensuremath{\text{\upshape BPL}}\xspace}

\newcommand{\traceterm}{\ensuremath{\mathsf{t}_\mathsf{tr}}\xspace}

\newcommand{\SAINFER}[3]{
\begin{prooftree}
\AxiomC{\phantom{$E^E$}}
\LeftLabel{\rulename{#2}}
\RightLabel{#1}
\UnaryInfC{#3}
\end{prooftree}
}

\newcommand{\symbolic}[1]{\ensuremath{\underline{#1}}}
\newcommand{\pairdelim}{\genfrac{(}{)}{0pt}{}}
\newcommand{\statepair}[2]{\ensuremath{\pairdelim{#1}{#2}}\xspace}

\newcommand{\defstatepair}{\statepair{\state}{\heap}}
\newcommand{\defstatepairo}{\statepair{\state'}{\heap'}}

\newcommand{\marker}{\ensuremath{\diamond}\xspace}

\newcommand{\destiny}{\ensuremath{f}\xspace}

\newcommand{\process}{\ensuremath{\mathsf{proc}}\xspace}
\newcommand{\processpool}{\ensuremath{\mathsf{procpool}}\xspace}

\newcommand{\defeval}[1]{\ensuremath{ \llbracket #1 \rrbracket_{\objname,\destiny,\methodname,\defstatepair}\xspace}}
\newcommand{\anyeval}[1]{\ensuremath{ \left\llbracket #1 \right\rrbracket \xspace}}
\newcommand{\eval}[2]{\ensuremath{ \left\llbracket #1 \right\rrbracket_{#2} \xspace}}

%% file: chapters/03_active.tex
\section{\coreactor: An Active Object Language with Locally Abstract, Globally Concrete Semantics}\label{ch:ao}
We use the Active Object concurrency model, an extension of the Actor~\cite{actor} concurrency model.
Actors are endpoints which communicate solely via asynchronous message-passing.
An actor may receive messages at any time and each message is handled \emph{preemption-free}: once the handler for a message has started, it may not be interrupted by any other handler.
Messages arriving during execution of a handler are stored in a bag data structure.

As Actors do not share memory and handlers cannot be interrupted, reasoning about concurrent behavior in Actor-based systems is simpler than for C-style concurrency models with shared memory, where two process may interfere at any point, except when explicitly forbidden to do so.
Indeed, compared to cloud systems based on other paradigms, real-world systems build upon pure actors sustain less coordination-based bugs~\cite{HeddenZ18}.
Erlang~\cite{Armstrong07} and the Akka framework are some mainstream implementations of Actors, with library-support for further languages being available.

In an object-oriented setting, message sending can be encoded asynchronous method calls and the handlers as the corresponding methods.
Object-orientation also adds a further possibility for interaction: processes within one Actor may communicate by storing and reading data in the objects \emph{heap memory}.
However, there is still no preemption, so the points where there interactions happen are clear. 

Active Object languages are object-oriented languages that enforce the Actor concurrency model
Additionally, modern Active Object languages such as ABS~\cite{abs}, Encore~\cite{BrandauerCCFJPT15} or ProActive~\cite{ProActive} use \emph{futures}~\cite{BakerH77,future}.
A future is an identifier for a method call that allows to synchronize on the called process and read its return value. 
Futures simplify the sending of a result back to the caller Actor. 
Instead of sending it to a second method via a callback, a downside of the pure Actor model\COMMENT{ informally called ``callback hell''}, the result may now be handled in the same method by synchronizing on the called process. 
However, futures enable Active Objects to deadlock by circular synchronization on futures.

Finally, Creol~\cite{JohnsenOY06} and ABS introduced \emph{cooperative scheduling}: a process may be descheduled, but only at a special \abs{await} statement. The statements between two such suspension points still have exclusive access to the objects heap memory,
data races may only occur at method start and \abs{await} statements.

We define \coreactor in this chapter, an Active Object language with futures, consequent strong encapsulation (i.e., all fields are object-private) and cooperative scheduling. 
For a further discussion of Active Object languages we refer to the survey of de~Boer et al.~\cite{BoerSHHRDJSKFY17}.

\subsection{Syntax}
\begin{definition}[Syntax]\label{def:syntax}
Let \xabs{v} range over program variables, \xabs{f} over field names, \classname over class names, \methodname over method names, $i$ and \abs{n} over \(\mathbb{N}\) and $\sim$ range over $\xabs{&&},\xabs{||},\xabs{+},\xabs{-},\xabs{*},\xabs{/}$. The syntax of \coreactor is defined by the grammar in Fig.~\ref{fig:syntax}. We use $\many{\cdot}$ to denote (possibly empty) lists.
\begin{figure}[h]
\begin{align*}
\prgm ::=&~ \many{\classlang}~\mainlang
\quad \mainlang ::=~\mainABS\{\synt{si}\} 
\quad \classlang ::=~\xabs{class}~\classname~\xabs{(}\many{\classname~\xabs{f}}\xabs{)}\{\many{\fieldlang}~\many{\methodlang}\}\\
\methodlang ::=&~\dType~\methodname(\many{\dType~\varname})\{\synt{s}\}
\quad\fieldlang ::=~\dType~\xabs{f = e;}
\quad\synt{si} ::= \synt{C v = C(}\many{\synt{v}}\synt{); si} \sep \synt{v!m(}\many{\expr}\synt{)}\\
\dType ::=&~ \xabs{Rat} \sep \xabs{Unit} \sep \xabs{Int} \sep\xabs{List<}\dType\xabs{>} \sep \xabs{Bool} \sep \xabs{Fut<}\dType\xabs{>}\\
\synt{s} ::=&~
[\dType]~ \synt{l = e} \sep 
[\dType]~\synt{v = e.get}_i\sep 
\synt{await g}_i\sep 
[\dType]~\xabs{v = f!m(}\many{\xabs{e}}\xabs{)} \sep 
\synt{skip} \sep 
\xabs{return e} \\
&\sep
\xabs{while(e)}\{\xabs{s}\} \sep
\xabs{if(e)} \{\xabs{s}\}\xabs{else}\{\xabs{s}\}\sep
\synt{s;s}\\
\explang ::=&~ \xabs{l} \sep \xabs{n} \sep \xabs{Never} \sep \xabs{Nil} \sep \xabs{True} \sep \xabs{False} \sep \xabs{len(}\explang\xabs{)} \sep \xabs{hd(}\explang\xabs{)} \sep \xabs{tl(}\explang\xabs{)} \sep \xabs{Cons(}\explang\xabs{,}\explang\xabs{)} \\
&\sep \explang \sim \explang  \sep\xabs{!}\explang \sep \xabs{-}\explang \qquad \xabs{l} ::=~ \xabs{this.f}\sep\xabs{v}\qquad \xabs{g} ::=~ \explang \sep \explang?
\end{align*}
\caption{Syntax of \coreactor.}
\label{fig:syntax}
\end{figure}
\end{definition}
A \coreactor program consists of a set of classes and a main block.
The main block is a list of object instantiantions and one final method call to one of the instantiantions to initialize the communication.
The object instantiantions bind the created object to a name and take as parameters references to other objects. 
All objects are created at once and thus the order of object creations is not important.
E.g., the following initializes \abs{a} and \abs{b} with pointers to each other.

\noindent\begin{abscode}[numbers=none]
main{ C a = C(b); C b = C(a); a!m(); }
\end{abscode}

A class has (1) a list of references to other objects (2) a list of fields which are initialized upon creation and (3) a list of methods.
The references are regarded as fields that cannot be reassigned.
\coreactor does not support inheritance.\COMMENT{no: aliasing\ektodo{maybe one can allows classes to be data types?}} 
There are six kinds of data types: Rationals, the Unit type, Integers, Booleans, parametric lists and parametric futures.
The expressions for Booleans and Integers are straightforward. \abs{no} is a unique future that is never resolved, \abs{Nil} is the empty list, \abs{hd} returns the first element of a list and \abs{tl} the remainder.
\abs{Cons} is the list constructor and \abs{len} returns the length of a list.

The statements for assignment, branching, repetition and the empty statement are standard.
The statement \abs{v = f!m(}$\many{\xabs{e}}$\xabs{)} is an asynchronous method call on method \abs{m} on the object \abs{f} with parameters $\many{\xabs{e}}$. 
A fresh future is generated and stored in \abs{v}. This future identifies the called process. 
We say that the called process will \emph{resolve} this future.
The statement \abs{return e} terminates the process and stores the value \abs{e} in the future it is resolving. 
The statement \abs{v = e.get} synchronizes on the future in \abs{e}. Once this future is resolved, the result stored in it is written into \abs{v}. Until the future is resolved, the process blocks the object.
The statement \abs{await g} suspends the current process until the guard becomes active.
A future guard \abs{e?} becomes active, once the future stored in \abs{e} becomes active. A boolean guard \abs{e} becomes active once the boolean expression \abs{e} evaluates to true.
Once the guard is active, the process \emph{may} be rescheduled, but if several process may be scheduled, one is chosen non-deterministically.
The \abs{await} and \abs{get} statements have an program-point identifier $i$. This identifier is used later by the specification to distinguish multiple suspension and synchronization points.

Methods are standard, but parameters cannot be reassigned.
We only consider \coreactor programs which fulfill the following constraints:
\begin{itemize}
\item They are data type checked.\footnote{We do not give a data type system: the system from~\cite{abs} is applicable for the classes and checking the main block is trivial.}
\item Each method has exactly one \abs{return} statement, which is the very last statement of the method body.
\item Each branch of the \abs{if} statement and the loop body of \abs{while} end in a \abs{skip} statement. 
\item Each variable name and field name and program-point identifier is program-wide unique. 
\end{itemize}

We drop empty \abs{else} branches and (if not needed) the program-point identifiers in examples.
We assume that every branch ends in an \abs{skip} statement.
\COMMENT{
\begin{example}
The code in Fig.~\ref{ex:ema} sets up an server \abs{ema} that computes the exponential moving average with $\alpha=0.5$ and a client \abs{c} sending values to it.
If the moving average drops below 0, the server interrupts the client. In the setup example, this happens once the value \xabs{-4} is send.

\begin{figure}
\begin{minipage}{0.45\textwidth}
\begin{abscode}
class Client (EMAServer server){
  Bool stopped = False;
  Unit trigger(){ 
      this.stopped = True; 
  }
  Unit compute(List<Int> values){
    List<Int> val = values;
    while(val != Nil){
      Int v = hd(val);
      Fut<Rat> f = server.ema(v);
      await f?;
      val = tl(val);
      if(this.stopped) {val = Nil;}
    }
  }
}
\end{abscode}
\end{minipage}
\begin{minipage}{0.55\textwidth}
\begin{abscode}[firstnumber=17]
class EMAServer (Client c){
  Rat val = 1; Rat a = 1/2;
  Rat ema(Int next){
    this.val = this.a*next+(1-this.a)*this.val;
    if(this.val < 0) {
      Fut<Unit> f = c!trigger();
      Unit unit = f.get;
    }
    return this.val;
  }
}
main{
  EMAServer ema = new EMAServer(c);
  Client c = new Client(ema);
  c!compute(Cons(2,Cons(3,Cons(-4,Cons(5,Nil)))));
}
\end{abscode}
\end{minipage}
\caption{Computing the Exponential Moving Average in \coreactor.}
\label{ex:ema}
\end{figure}

\end{example}
}

\subsection{Locally Abstract Semantics}
\input{chapters/031_local}

\subsection{Globally Concrete Semantics}
\input{chapters/032_global}

\subsection{Selectability}
\input{chapters/035_select}

\subsection{Semantic Logics}\label{sec:aol}
\input{chapters/033_logic}


%% file: chapters/031_local.tex
We use a locally abstract, globally concrete (LAGC) semantics~\cite{DinHJPT17} for \coreactor.
A LAGC semantics consists of two layers: A locally abstract (LA) layer for statements and methods, and a globally concrete (GC) layer for objects and systems. 
The LA layer is a denotational semantics that abstractly describes the behavior of a method in every possible context, while the GC layer is a operational semantics that concretize the 
abstract behavior of the processes in a concrete context.
LAGC semantics allow to analyze a method in isolation and to express possible context by concretizing the most abstract behavior bit by bit. 

The semantics of a method is a set of symbolic traces, which all together describe the behavior of the method in every possible context, i.e., for every possible object heap, call parameters and accessed futures.
Symbolic traces contain symbolic values and symbolic expressions, additionally to semantic values (a semantics value is e.g., an integer stored in a variable of \abs{Int} type).
Symbolic values have no operations defined on them, instead they act as placeholders. 
They are replaced by semantic values once the method is running and the object heap, call parameters etc. are known.

We first define the semantic values.
Given a program \prgm with $n$ object creations in its main block, the set $(\objname_i)_{i\leq n}$ is the set of object names. $\objname_i$ is the object name of the $i$th object creation.
\begin{definition}[Semantic Values]
The set of \emph{semantic values} contains the unit $\mathbf{unit}$, the boolean values $\mathbf{True},\mathbf{False}$, the natural numbers $n$, all object names and all future identities, including $\mathbf{no}$. 
Additionally, a sequence of semantic values is also a semantic value.
\end{definition}
Futures are identities, not containers. We say that a future contains a value, but this is established by the traces, not by some explicit state of the future.

\subsubsection{Semantics of Expressions}
Symbolic expressions contain symbolic values and semantic values, but while their structure mirrors the syntax of \coreactor-expressions, they do not contain any reference to state, i.e., no variables, fields or constants.
Instead, they contain symbolic values and symbolic fields. Symbolic fields act like symbolic values, but contain the name of the field they are abstracting and a counter. This is need to later substitute a concrete value in a trace up to the next point where the heap may change.
\begin{definition}[Symbolic Values and Expressions]
A symbolic expression \sexplang is defined by the below grammar, analogously to the grammar of Def.~\ref{def:syntax} above, with semantic values $v$, symbolic values $\symbolic{v}$ and symbolic fields $\symbolic{\xabs{this.f}}_i$ instead of variables, fields or constants.
\[\sexplang ::= \xabs{len}\xabs{(}\sexplang\xabs{)} \sep {\xabs{hd}}\xabs{(}\sexplang\xabs{)} \sep 
                \xabs{tl}\xabs{(}\sexplang\xabs{)} \sep {\xabs{Cons}}\xabs{(}\sexplang\xabs{,}\sexplang\xabs{)} \sep \sexplang \sim \sexplang \sep \xabs{!}\sexplang \sep \xabs{-}\sexplang \sep \symbolic{v} \sep v \sep \symbolic{\xabs{this.f}}_i\]
Where $i$ ranges over $\mathbb{N}$.
To clearly separate symbolic values and fields from other expressions, we always denote them underlined, e.g., $\symbolic{v}$ is a symbolic value and $\symbolic{v}+1$ is a symbolic expression.
In contrast, \xabs{1+1} is a syntactic expression and $2$ a semantic value.
\end{definition}
We stress that semantic values are a subset of symbolic expressions. Syntactic expressions and symbolic expressions share only their operators.

The local state of a method and the heap of a method are functions from variable names (resp. field names) to symbolic expressions.
\begin{definition}[Local State and Heap]
Given a class \classname, a heap \heap is a function from the fields of \classname to symbolic expressions.
Given a method \classname.\methodname, a local state \state is a map from the variables of \classname.\methodname to symbolic expressions.
A local state or heap is \emph{concrete} if all elements of its image are values. We call a pair of a local state \state and a heap \heap a \emph{object state} and write it \statepair{\state}{\heap}. If the local state is empty, we write \statepair{\cdot}{\heap}.

We denote the heap that maps every field \xabs{f} to the symbolic expression $\symbolic{\xabs{this.f}}_i$ and every parameter according to $\heap$ with $\heapid^i(\heap)$.
\end{definition}

We are now able to give a semantics to syntactic expressions, by evaluating them under a given local state and a heap to a symbolic expression.
\begin{definition}[Semantics of Expressions]
The evaluation $\anyeval{\expr}$ on semantic values (and operations on them) has its natural definition and we identify the operators of $\sim$ with their natural counterpart.
Let \expr be an expression and $\defstatepair$ an object state. The evaluation function $\eval{\expr}{\defstatepair}$ is defined in Fig.~\ref{fig:semexpr}.
The evaluation function is not defined if division by zero occurs or the head of an empty list is accessed.
\begin{figure}
\noindent\begin{minipage}{0.5\textwidth}
\begin{align*}
\eval{\xabs{len(e)}}{\statepair{\state}{\heap}} &= \cased{ \left|\eval{\expr}{\statepair{\state}{\heap}}\right| & \text{if \eval{\expr}{\statepair{\state}{\heap}} is a semantic value} \\ \xabs{len(}|\eval{\expr}{\statepair{\state}{\heap}}\xabs{)} & \text{if \eval{\expr}{\statepair{\state}{\heap}} is symbolic}}\\
\eval{\xabs{-e}}{\statepair{\state}{\heap}} &= \cased{ \anyeval{-\eval{\expr}{\statepair{\state}{\heap}}} & \text{if \eval{\expr}{\statepair{\state}{\heap}} is a semantic value} \\ \eval{\expr}{\statepair{\state}{\heap}} &\text{if \eval{\expr}{\statepair{\state}{\heap}} is symbolic}}\\
\eval{\xabs{!e}}{\statepair{\state}{\heap}} &= \cased{ \anyeval{!\eval{\expr}{\statepair{\state}{\heap}}} & \text{if \eval{\expr}{\statepair{\state}{\heap}} is a semantic value} \\ \eval{\expr}{\statepair{\state}{\heap}} &  \text{if \eval{\expr}{\statepair{\state}{\heap}} is symbolic}}
\end{align*}
\end{minipage}
\begin{minipage}{0.5\textwidth}
\begin{align*}
\eval{\xabs{this.f}}{\statepair{\state}{\heap}} &= \heap(\xabs{f}) \\
\eval{\xabs{v}}{\statepair{\state}{\heap}} &= \state(\xabs{v}) \\
\eval{\xabs{n}}{\statepair{\state}{\heap}} &= n \\
\eval{\xabs{Nil}}{\statepair{\state}{\heap}} &= \sequence{} \\
\eval{\xabs{True}}{\statepair{\state}{\heap}} &= \mathbf{True} \\
\eval{\xabs{False}}{\statepair{\state}{\heap}} &= \mathbf{False} \\
\eval{\xabs{Cons(}\expr\xabs{,}\expr')}{\statepair{\state}{\heap}} &= \sequence{\eval{\expr}{\statepair{\state}{\heap}}} \circ\eval{\expr'}{\statepair{\state}{\heap}}
\end{align*}
\end{minipage}
\begin{align*}
\eval{\xabs{hd(e)}}{\statepair{\state}{\heap}} &= \cased{ \xabs{hd(}\eval{\xabs{e}}{\statepair{\state}{\heap}}\xabs{)} & \text{if \eval{\expr}{\statepair{\state}{\heap}} is symbolic} \\ \eval{\xabs{e}}{\statepair{\state}{\heap}}\left[1\right] & \text{if \eval{\expr}{\statepair{\state}{\heap}} is a semantic value and not empty}}\\
\eval{\xabs{tl(e)}}{\statepair{\state}{\heap}} &= \cased{ \xabs{tl(}\eval{\xabs{e}}{\statepair{\state}{\heap}}\xabs{)} & \text{if \eval{\expr}{\statepair{\state}{\heap}} is symbolic} \\ \eval{\xabs{e}}{\statepair{\state}{\heap}}\left[2..\left|\eval{\xabs{e}}{\statepair{\state}{\heap}}\right|\right] & \text{if \eval{\expr}{\statepair{\state}{\heap}} is a semantic value and not empty}}\\
\eval{\expr \sim \expr'}{\statepair{\state}{\heap}} &= \cased{ \anyeval{\eval{\expr}{\statepair{\state}{\heap}} \sim \eval{\expr'}{\statepair{\state}{\heap}}} &\text{if \eval{\expr}{\statepair{\state}{\heap}}, \eval{\expr'}{\statepair{\state}{\heap}} are semantic and $\eval{\expr'}{\statepair{\state}{\heap}} \neq0$}\\ \eval{\expr}{\statepair{\state}{\heap}} \sim \eval{\expr'}{\statepair{\state}{\heap}} &\text{if either \eval{\expr}{\statepair{\state}{\heap}} or \eval{\expr'}{\statepair{\state}{\heap}} are symbolic}}
\end{align*}
\caption{Locally Abstract Semantics of Expressions.}
\label{fig:semexpr}
\end{figure}
\end{definition}

The semantics of expressions is well-defined, i.e., given a syntactic expression the semantics produce a symbolic expression.
Additionally, the semantics ensure that if a symbolic expression contains no symbolic value, then it is fully evaluated to a semantic value.
\begin{lemma}
Let \expr be an expression, \state a local state and \heap a heap, such that $\sexpr = \eval{\expr}{\statepair{\state}{\heap}}$ is defined
\begin{enumerate}
\item If \sexpr contains no symbolic value, then it is a semantic value. 
\item If \state and \heap are both concrete, then \sexpr is a semantic value.
\end{enumerate}
\end{lemma}

\begin{example}
Consider the expression \abs{this.f+i+hd(Cons(2,Nil))}, and the following states and heap for a class with no parameters and a field \abs{f} of \abs{Int} data type:
\\\noindent\scalebox{0.8}{\begin{minipage}{\textwidth}
\begin{align*}
\state &= \{\xabs{i} \mapsto 1\}\qquad\state' = \{\xabs{i} \mapsto\symbolic{i}\}\qquad\heap = \{\xabs{f} \mapsto 2\}\qquad \heap'=\heapid^1\\
\eval{\xabs{this.f+i+hd(Cons(2,Nil))}}{\statepair{\state}{\heap}} &= \eval{\xabs{this.f}}{\statepair{\state}{\heap}} + \eval{\xabs{i}}{\statepair{\state}{\heap}}  + \eval{\xabs{hd(Cons(2,Nil))}}{\statepair{\state}{\heap}} \\
&=2 + 1 + \left(\sequence{2}\circ\sequence{}\right)[1] = 5\\
\eval{\xabs{this.f+i+hd(Cons(2,Nil))}}{\statepair{\state}{\heap'}} &= \eval{\xabs{this.f}}{\statepair{\state}{\heap'}} + \eval{\xabs{i}}{\statepair{\state}{\heap'}}  + \eval{\xabs{hd(Cons(2,Nil))}}{\statepair{\state}{\heap'}} \\
&= \symbolic{\xabs{this.f}}_1 + 1 + \left(\sequence{2}\circ\sequence{}\right)[1] = \symbolic{\xabs{this.f}}_1 + 3\\
\eval{\xabs{this.f+i+hd(Cons(2,Nil))}}{\statepair{\state'}{\heap'}} &= \eval{\xabs{this.f}}{\statepair{\state'}{\heap'}} + \eval{\xabs{i}}{\statepair{\state'}{\heap'}}  + \eval{\xabs{hd(Cons(2,Nil))}}{\statepair{\state'}{\heap'}} \\
&= \symbolic{\xabs{this.f}}_1 + \symbolic{i} + \left(\sequence{2}\circ\sequence{}\right)[1] = \symbolic{\xabs{this.f}}_1 + \symbolic{i} + 2
\end{align*}
\end{minipage}
}
\end{example}

\subsubsection{Semantics of Statements and Methods}
The traces in the semantics of a method, contain not only states, but also events: markers for visible communication (and the event \noev for uniformity).
In the globally concrete semantics, the events are used to merge the local traces of a method first into the local trace of an object and then into a global trace of the whole system.
\begin{definition}[Events]
Fig.~\ref{fig:events} shows the syntax of events.
Each type of event is paired with a reaction event, which either models that the communication has had its effect on the other party, or that the process has regained control.
\begin{itemize}
\item
The invocation event $\invocev(\symbolic{\objname},\symbolic{\objname'},\symbolic{f},\methodname,\many{\sexpr})$ models a call from \objname  to $\symbolic{\objname'}$ on method \methodname. The future $\symbolic{f}$ is used and $\many{\sexpr}$ are the call parameters.
The future and the callee may be symbolic, because without global information it is not possible to know what future will be used, what object will be called and what object the local semantics are generated for.
\item 
The invocation reaction event $\invocrev(\symbolic{\objname},\symbolic{\future},\methodname,\many{\sexpr})$ is the callee view on the method call. The object \objname here is the callee, the caller is not visible to the callee.
\item 
The resolving event $\resolvev(\symbolic{\objname},\symbolic{\future},\methodname,\sexpr)$ models the termination of a process for future $\symbolic{\future}$ that computer method \methodname in object $\symbolic{\objname}$ and returns $\sexpr$.
\item 
The resolving reaction event $\resolvrev(\symbolic{\objname},\symbolic{\future}, \symbolic{\methodname},\sexpr, i)$ models the read of a \abs{get}$_i$ statement in object $\symbolic{\objname}$ on the future $\symbolic{\future}$, reading value $\sexpr$.
\item 
The condition synchronization event $\awaitev(\symbolic{\objname},\symbolic{\future},\expr, i)$ models that the process computing the future $\symbolic{\future}$ in object $\symbolic{\objname}$ on statement $\xabs{await e}_i$. Note that \expr is not symbolic, but the syntactic guard of the statement. 
\item
The condition synchronization reaction event $\awaitrev(\symbolic{\objname},\symbolic{\future},\expr, i)$ is analogous for continuing execution.
\item 
The suspension event $\suspev(\symbolic{\objname},\symbolic{\future},\sexpr, i)$ models that the process computing the future $\symbolic{\future}$ in object $\symbolic{\objname}$ on statement $\xabs{await e?}_i$, where $\sexpr$ is the future which is read.
\item 
The suspension reaction event is analogous to the condition synchronization reaction event.
\item 
Finally, \noev models a step without visible communication.
\end{itemize}
We say that $\invocev(\symbolic{\objname},\symbolic{\objname'},\symbolic{f},\methodname,\many{\sexpr})$ introduces \symbolic{f} and that $\resolvrev(\symbolic{\objname},\sexpr, \sexpr', i)$ introduces $\sexpr'$.
\begin{figure}
\begin{align*}
\event ::= \quad
&\invocev(\symbolic{\objname},\symbolic{\objname'},\symbolic{f},\methodname,\many{\sexpr}) \sep 
\invocrev(\symbolic{\objname},\symbolic{\future},\methodname,\many{\sexpr}) \sep
\resolvev(\symbolic{\objname},\symbolic{\future},\methodname,\sexpr) \sep
\resolvrev(\symbolic{\objname},\symbolic{\future}, \symbolic{\methodname},\sexpr, i)\\
\sep&\awaitev(\symbolic{\objname},\symbolic{\future},\expr, i) \sep
\awaitrev(\symbolic{\objname},\symbolic{\future},\expr, i) \sep
\suspev(\symbolic{\objname},\symbolic{\future},\sexpr, i) \sep
\susprev(\symbolic{\objname},\symbolic{\future},\sexpr, i) \sep
\noev
\end{align*}
\caption{Syntax of Events.}
\label{fig:events}
\end{figure}
\end{definition}
In the globally concrete semantics of a program, the object parameters of the events, as well as the computed future, are not symbolic. They are, however, symbolic when we analyze a method in isolation.

Local traces consist of two parts. (1) A selection condition, a set of symbolic expressions that express when this trace can be selected to execute the next step and (2) a history, a sequence of events and local states.
\begin{definition}[Local Traces and Chops]
A local trace $\trace$ has the form $\selection \triangleright \history$, where $\selection$ is a set of symbolic expressions, called \emph{selection condition}, and $\history$ is a non-empty sequence, called \emph{history}, such that 
\begin{itemize}
\item $|\history|$ is odd and or for all odd $i\leq|\history|$, $\history[i]$ is an object state 
\item For all even $i\leq|\history|$, $\history[i]$ is an event or the special symbol \marker
\end{itemize}
We use two chopping operators. The standard $\chop$ merges two histories if they share the last and first object state.
The extended chop $\megachop$ operates like $\chop$ if the standard chop is defined. If only the heaps are equal, it concatenates the histories with a marker \marker.
Both are not defined if the heaps are different. Both chops are lifted to traces by chopping the histories and joining the selection conditions.
\begin{align*}
\history\circ\sequence{\defstatepair} \chop \sequence{\statepair{\state'}{\heap'}}\circ\history' ~=~& \cased{\history \circ \sequence{\defstatepair} \circ \history' &\text{ if }\defstatepair = \statepair{\state'}{\heap'}\\ \text{undefined} &\text{ otherwise }}\\ 
\history\circ\sequence{\defstatepair} \megachop \sequence{\statepair{\state'}{\heap'}}\circ\history' ~=~& \cased{\history \circ \sequence{\defstatepair} \circ \history' &\text{ if }\defstatepair = \statepair{\state'}{\heap'}\\ 
\history \circ \sequence{\defstatepair,\marker,\statepair{\state'}{\heap'}} \circ \history' &\text{ if }\heap = \heap',\state \neq \state'\\
 \text{undefined} &\text{ otherwise }}\\
\selection \rhd \history \chop \selection' \rhd \history' ~=~& \selection \cup \selection' \rhd \history\chop\history'\\
\selection \rhd \history \megachop \selection' \rhd \history' ~=~& \selection \cup \selection' \rhd \history\megachop\history'
\end{align*}
\end{definition}

The marker models that at this point in the trace, another process may run, but \marker is no event.
Before defining the semantics of statements and methods, we require some technical definitions.
\begin{definition}[Freshness and Symbolic State]
A symbolic value or a heap counter $i$ is fresh, if it occurs nowhere else\footnote{To be more precise: nowhere in the state of the system, which we define in the next section. We refrain from introducing local and global freshness, as it offers no insights into LAGC semantics.}.
A heap (resp. local state) is symbolic if it maps some field (resp. variable) to a symbolic expression that is not a semantic value.
\end{definition}

\begin{definition}[Local Semantics]
The semantics of methods and statements is defined by a function $\defeval{\cdot}$ where $\objname$ is the object name, $\future$ the future the method is resolving, \methodname the method name and \defstatepair the current object state.
Future, object name and state may be symbolic.
The semantic of a method \methodname with method body \statement is
\[\defeval{\methodname} = \left\{\emptyset \triangleright \sequence{\defstatepair,\invocrev(\objname,\future,\methodname,\many{\expr}),\defstatepair} \chop \trace \sep \trace \in \defeval{\statement}\right\} \]
Where $\many{\expr}$ is extracted from the parameter names in $\state$. E.g., if the signature of a method is \abs{Int m(Int a, Rat b)} then $\many{\expr} = \sequence{\state(\xabs{a}),\state(\xabs{b})}$.
Fig.~\ref{fig:lasem} shows the rules for statements. We stress that if the semantics of an expression is not defined, then so is the semantics of a statement.
\begin{figure}
\scalebox{0.7}{
\begin{minipage}{\textwidth}
\begin{align*}
\defeval{\xabs{skip}} &= \left\{ \emptyset \triangleright \sequence{\defstatepair} \right\}\\
\defeval{\xabs{skip;s}} &= \defeval{\xabs{s}}\\
\defeval{\xabs{v = e;s}} &= \left\{ \emptyset \triangleright \sequence{\defstatepair,\noev,\statepair{\state'}{\heap}} \chop \trace \middle| \trace \in \eval{\xabs{s}}{\objname,\destiny,\methodname,\statepair{\state'}{\heap}}\right\}\text{ where $\state' = \state[\xabs{v} \mapsto \eval{\xabs{e}}{\defstatepair}]$}\\
\defeval{\xabs{D v = e;s}} &= \defeval{\xabs{v = e;s}}\\
\defeval{\xabs{this.f = e;s}} &= \left\{ \emptyset \triangleright \sequence{\defstatepair,\noev,\statepair{\state}{\heap'}} \chop \trace \middle| \trace \in \eval{\xabs{s}}{\objname,\destiny,\methodname,\statepair{\state}{\heap'}}\right\}\text{ where $\heap' = \heap[\xabs{f} \mapsto \eval{\xabs{e}}{\defstatepair}]$}\\
\defeval{\xabs{return e}} &= \left\{ \emptyset \triangleright \sequence{\defstatepair,\resolvev \left(\objname,\destiny,\methodname,\eval{\xabs{e}}{\statepair{\state}{\heap}}\right),\defstatepair} \right\}\\
\defeval{\xabs{if(e)}\{\xabs{s}\}\xabs{else}\{\xabs{s'}\}\xabs{;s''}} &= ~ ~ ~
\left\{ \{\eval{\xabs{e}}{\defstatepair}\} \triangleright \sequence{\defstatepair}\chop\trace\middle| \trace \in \eval{\xabs{s;s''}}{\objname,\destiny,\methodname,\defstatepair} \right\}\\
&~\cup~ \left\{ \{\eval{\xabs{!e}}{\defstatepair}\} \triangleright \sequence{\defstatepair}\chop\trace\middle| \trace \in \eval{\xabs{s';s''}}{\objname,\destiny,\methodname,\defstatepair} \right\}
\end{align*}
\begin{align*}
\defeval{\xabs{v = e.get}_i\xabs{;s}} &= \\ &\hspace{-25mm}\left\{  \emptyset \triangleright 
\sequence{\defstatepair,\resolvrev\left(\objname,\eval{\xabs{e}}{\defstatepair},\symbolic{v},i\right) ,\statepair{\state'}{\heap}}\chop \trace 
\middle| \trace \in \eval{\xabs{s}}{\objname,\destiny,\methodname,\statepair{\state'}{\heap}}\right\} \text{where  $\state' = \state[\xabs{v} \mapsto \symbolic{v}]$} \\
\defeval{\xabs{v = v'!n(}\many{\xabs{e}}\xabs{);s}} &= \\ &\hspace{-25mm}\left\{  \emptyset \triangleright 
\sequence{\defstatepair,\invocev\left(\objname,\eval{\xabs{v}}{\defstatepair},\symbolic{v},\xabs{n},\many{\eval{\xabs{e}}{\defstatepair}}\right) ,\statepair{\state'}{\heap}}\chop \trace 
\middle| \trace \in \eval{\xabs{s}}{\objname,\destiny,\methodname,\statepair{\state'}{\heap}}\right\} \text{where $\state' = \state[\xabs{v} \mapsto \symbolic{v}]$} \\
\defeval{\xabs{await e?}_i\xabs{;s}} &= \\ &\hspace{-25mm}\left\{  \emptyset \triangleright 
\sequence{\defstatepair,\suspev\left(\objname,\destiny,\eval{\xabs{e}}{\defstatepair},i\right),\defstatepair,\marker,\statepair{\state}{\heapid^j(\heap)},\susprev\left(\objname,\destiny,\eval{\xabs{e}}{\defstatepair},i\right),\statepair{\state}{\heapid^i(\heap)}}\chop \trace 
\middle| \trace \in \eval{\xabs{s}}{\objname,\destiny,\methodname,\statepair{\state}{\heapid^j(\heap)}}\right\}\\
\defeval{\xabs{await e}_i\xabs{;s}} &= \\ &\hspace{-25mm}\left\{  \emptyset \triangleright 
\sequence{\defstatepair,\awaitev\left(\objname,\destiny,\xabs{e},i\right),\defstatepair,\marker,\statepair{\state}{\heapid^j(\heap)},\awaitrev\left(\objname,\destiny,\exp,i\right),\statepair{\state}{\heapid^i(\heap)}}\chop \trace 
\middle| \trace \in \eval{\xabs{s}}{\objname,\destiny,\methodname,\statepair{\state}{\heapid^j(\heap)}}\right\}\\
\defeval{\xabs{while(e)}\{\xabs{s}\}\xabs{;s'}} &= \defeval{\xabs{if(e)}\{\xabs{s};\xabs{while(e)}\{\xabs{s}\}\xabs{;skip}\}\xabs{else}\{\xabs{skip}\}\xabs{;s'}} 
\end{align*}
\end{minipage}
}
\caption{Locally Abstract Semantics of Statements, where the symbolic $\symbolic{v}$ value and the number $j$ are fresh in their rules.}
\label{fig:lasem}
\end{figure}
\begin{itemize}
\item
The first rule for \abs{skip} generates a trace that can always be selected and consists solely of the current state. 
If \abs{skip} is followed by some statement \abs{s}, all traces of its semantic are used.
\item 
The rule for assignment of variables updates the local state, adds a \noev event and is followed by the traces of the following statement, generated with the updated store.
\item 
The rule for variable declaration just skips over the declaration and only evaluates the initialization. The local state has all local variables in its domain, the type system ensures that the variables are not used before declaration.
\item 
The rule for assignment of fields is analogous, but updates the heap instead of the local state.
\item 
The rule for branching adds the guard to the selection condition of the first branch and its negation to the selection condition of the second.
\item 
The rule for \abs{get} is similar to the variable assignment, but gets a fresh symbolic value and stores it in the local state. As the event, and resolving reaction event is added, which stores the accessed future and the fresh symbolic value.
Note that the accessed future may also be symbolic.
\item 
The rule for method calls is analogous, but uses a fresh future for the call instead of a fresh read value. The added event is an invocation event with the evaluated parameters.
\item 
The two rules for \abs{await} are the only ones that add a marker \marker. In both rules, the next statement is evaluated with a fresh symbolic heap.
\item 
Finally, \abs{while} is defined recursively. 
This turns the used equations into a fixpoint definition, but we refrain from formally introducing fixpoints, because we only analyze terminating systems with finite traces which can be obtained by an unbounded, but finite number of applications of the \abs{while} rule. The formalism needed behind infinite traces can be found in~\cite{BubelDHN15}.
\end{itemize}
\end{definition}

\begin{example}\label{ex:lasem1}
Consider the following method .
\\\noindent\begin{minipage}{0.49\textwidth}
\begin{abscode}
Int m(Int p){
    this.f = this.f + 1; 
    Fut<Int> fut = o!n();
    Int i = fut.get;$_1$
    if(this.f > i){
        await this.f < 0;$_2$ 
\end{abscode}
\end{minipage}
\begin{minipage}{0.49\textwidth}
\begin{abscode}[firstnumber=7]
        if(this.f < p){
            this.f = 0;
        }
    }
    return this.f;
}
\end{abscode}
\end{minipage}

\noindent Its semantics for a call with parameter 5 in a class that has object $\objname'$ stored in its parameter \abs{o} are shown in Fig.~\ref{fig:exmsem}.
For readability's sake we omitted the parameters \abs{p} and \abs{o} from the object state, as parameters cannot be reassigned.

\begin{figure}
\noindent\scalebox{0.65}{
\begin{minipage}{\textwidth}
\begin{align*}
\trace_1 = &\{
    \symbolic{\xabs{this.f}}_1 > \symbolic{i}, \symbolic{\xabs{this.f}}_2 < 5\} \triangleright \\
    &\sequence{
    \statepair{\xabs{i} \mapsto 0, \xabs{fut} \mapsto \mathbf{no}}{\xabs{f} \mapsto \symbolic{\xabs{this.f}}_1},
    \invocrev(\objname,\destiny,\methodname,\sequence{5}),
    \statepair{\xabs{i} \mapsto 0, \xabs{fut} \mapsto \mathbf{no}}{\xabs{f} \mapsto \symbolic{\xabs{this.f}}_1},
    \noev,
    \statepair{\xabs{i} \mapsto 0, \xabs{fut} \mapsto \mathbf{no}}{\xabs{f} \mapsto \symbolic{\xabs{this.f}}_1+1},
    \invocev(\objname, \objname', \symbolic{f}, \xabs{n}, \sequence{})
    }\\
\circ &
\sequence{
    \statepair{\xabs{i} \mapsto 0, \xabs{fut} \mapsto \symbolic{f}}{\xabs{f} \mapsto \symbolic{\xabs{this.f}}_1+1},
    \resolvrev(\objname,\symbolic{f},\symbolic{i},1),
    \statepair{\xabs{i} \mapsto \symbolic{i}, \xabs{fut} \mapsto \symbolic{f}}{\xabs{f} \mapsto \symbolic{\xabs{this.f}}_1+1},
    \suspev,
    \statepair{\xabs{i} \mapsto \symbolic{i}, \xabs{fut} \mapsto \symbolic{f}}{\xabs{f} \mapsto \symbolic{\xabs{this.f}}_1+1},
    \marker
    }\\
\circ &
\sequence{
    \statepair{\xabs{i} \mapsto \symbolic{i}, \xabs{fut} \mapsto \symbolic{f}}{\xabs{f} \mapsto \symbolic{\xabs{this.f}}_2},
    \susprev,
    \statepair{\xabs{i} \mapsto \symbolic{i}, \xabs{fut} \mapsto \symbolic{f}}{\xabs{f} \mapsto \symbolic{\xabs{this.f}}_2},
    \noev,
    \statepair{\xabs{i} \mapsto \symbolic{i}, \xabs{fut} \mapsto \symbolic{f}}{\xabs{f} \mapsto 0},
    \resolvev(\objname,\destiny,\methodname,\symbolic{\xabs{this.f}}_2),
    \statepair{\xabs{i} \mapsto \symbolic{i}, \xabs{fut} \mapsto \symbolic{f}}{\xabs{f} \mapsto 0}
    }\\[5mm]
\trace_2 &= \{\symbolic{\xabs{this.f}}_1 > \symbolic{i}, \symbolic{\xabs{this.f}}_2 \geq 5\} \triangleright \\
&\sequence{
    \statepair{\xabs{i} \mapsto 0, \xabs{fut} \mapsto \mathbf{no}}{\xabs{f} \mapsto \symbolic{\xabs{this.f}}_1},
    \invocrev(\objname,\destiny,\methodname,\sequence{5}),
    \statepair{\xabs{i} \mapsto 0, \xabs{fut} \mapsto \mathbf{no}}{\xabs{f} \mapsto \symbolic{\xabs{this.f}}_1},
    \noev,
    \statepair{\xabs{i} \mapsto 0, \xabs{fut} \mapsto \mathbf{no}}{\xabs{f} \mapsto \symbolic{\xabs{this.f}}_1+1},
    \invocev(\objname, \objname', \symbolic{f}, \xabs{n}, \sequence{})
    }\\
\circ &
\sequence{
    \statepair{\xabs{i} \mapsto 0, \xabs{fut} \mapsto \symbolic{f}}{\xabs{f} \mapsto \symbolic{\xabs{this.f}}_1+1},
    \resolvrev(\objname,\symbolic{f},\symbolic{i},1),
    \statepair{\xabs{i} \mapsto \symbolic{i}, \xabs{fut} \mapsto \symbolic{f}}{\xabs{f} \mapsto \symbolic{\xabs{this.f}}_1+1}
    \suspev,
    \statepair{\xabs{i} \mapsto \symbolic{i}, \xabs{fut} \mapsto \symbolic{f}}{\xabs{f} \mapsto \symbolic{\xabs{this.f}}_1+1},
    \marker
    }\\
\circ &
\sequence{
    \statepair{\xabs{i} \mapsto \symbolic{i}, \xabs{fut} \mapsto \symbolic{f}}{\xabs{f} \mapsto \symbolic{\xabs{this.f}}_2},
    \susprev,
    \statepair{\xabs{i} \mapsto \symbolic{i}, \xabs{fut} \mapsto \symbolic{f}}{\xabs{f} \mapsto \symbolic{\xabs{this.f}}_2},
    \resolvev(\objname,\destiny,\methodname,\symbolic{\xabs{this.f}}_2),
    \statepair{\xabs{i} \mapsto \symbolic{i}, \xabs{fut} \mapsto \symbolic{f}}{\xabs{f} \mapsto \symbolic{\xabs{this.f}}_2}
    }\\[5mm]
\trace_3 &= \{\symbolic{\xabs{this.f}}_1 \leq \symbolic{i}\} \triangleright \\&\sequence{
    \statepair{\xabs{i} \mapsto 0, \xabs{fut} \mapsto \mathbf{no}}{\xabs{f} \mapsto \symbolic{\xabs{this.f}}_1},
    \invocrev(\objname,\destiny,\methodname,\sequence{5}),
    \statepair{\xabs{i} \mapsto 0, \xabs{fut} \mapsto \mathbf{no}}{\xabs{f} \mapsto \symbolic{\xabs{this.f}}_1},
    \noev,
    \statepair{\xabs{i} \mapsto 0, \xabs{fut} \mapsto \mathbf{no}}{\xabs{f} \mapsto \symbolic{\xabs{this.f}}_1+1},
    \invocev(\objname, \objname', \symbolic{f}, \xabs{n}, \sequence{}),
    }\\
\circ &
\sequence{
    \statepair{\xabs{i} \mapsto 0, \xabs{fut} \mapsto \symbolic{f}}{\xabs{f} \mapsto \symbolic{\xabs{this.f}}_1+1},
    \resolvrev(\objname,\symbolic{f},\symbolic{i},1),
    \statepair{\xabs{i} \mapsto \symbolic{i}, \xabs{fut} \mapsto \symbolic{f}}{\xabs{f} \mapsto \symbolic{\xabs{this.f}}_1+1},
    \resolvev(\objname,\destiny,\methodname,\symbolic{\xabs{this.f}}_1+1),
    \statepair{\xabs{i} \mapsto \symbolic{i}, \xabs{fut} \mapsto \symbolic{f}}{\xabs{f} \mapsto \symbolic{\xabs{this.f}}_1+1}
    }
\end{align*}
\end{minipage}}
\caption{Example Semantics.}
\label{fig:exmsem}
\end{figure}
\end{example}

\begin{definition}[Applying Heaps]
Let $\heap$ be a heap. We write $\history\heap$ for the history that results from applying \heap, i.e., replacing all symbolic fields $\symbolic{\xabs{this.f}}_i$ with $\heap(\xabs{f})$ up to the first $\marker$.
Given a trace $\trace$, we apply heaps with $\trace\heap = \selection\heap \triangleright \history\heap$.
The selection condition set $\selection\heap$ only substitutes for those symbolic fields $\symbolic{\xabs{this.v}}_i$ which were substituted in $\history\heap$.
We assume that resulting semantic values are evaluated directly after substitution.
\end{definition}

\begin{example}\label{ex:lasem2}
Consider the trace $\trace_1$ from Ex.~\ref{ex:lasem1} and the heap $\heap = \{\xabs{f} \mapsto 2 \}$.
\\\noindent\scalebox{0.65}{
\begin{minipage}{\textwidth}
\begin{align*}
\trace_1\heap &= 
    \{2 > \symbolic{i}, \symbolic{\xabs{this.f}}_2 < 5\} \triangleright \\&\sequence{
    \statepair{\xabs{i} \mapsto 0, \xabs{fut} \mapsto \mathbf{no}}{\xabs{f} \mapsto 2},
    \invocrev(\objname,\destiny,\methodname,\sequence{5}),
    \statepair{\xabs{i} \mapsto 0, \xabs{fut} \mapsto \mathbf{no}}{\xabs{f} \mapsto 2},
    \noev,
    \statepair{\xabs{i} \mapsto 0, \xabs{fut} \mapsto \mathbf{no}}{\xabs{f} \mapsto 3},
    \invocev(\objname, \objname', \symbolic{f}, \xabs{n}, \sequence{})
    }\\
\circ &
\sequence{
    \statepair{\xabs{i} \mapsto 0, \xabs{fut} \mapsto \symbolic{f}}{\xabs{f} \mapsto 3},
    \resolvrev(\objname,\symbolic{f},\symbolic{i},1),
    \statepair{\xabs{i} \mapsto \symbolic{i}, \xabs{fut} \mapsto \symbolic{f}}{\xabs{f} \mapsto 3},
    \suspev,
    \statepair{\xabs{i} \mapsto \symbolic{i}, \xabs{fut} \mapsto \symbolic{f}}{\xabs{f} \mapsto 3},
    \marker
    }\\
\circ &
\sequence{
    \statepair{\xabs{i} \mapsto \symbolic{i}, \xabs{fut} \mapsto \symbolic{f}}{\xabs{f} \mapsto \symbolic{\xabs{this.f}}_2},
    \susprev,
    \statepair{\xabs{i} \mapsto \symbolic{i}, \xabs{fut} \mapsto \symbolic{f}}{\xabs{f} \mapsto \symbolic{\xabs{this.f}}_2},
    \noev,
    \statepair{\xabs{i} \mapsto \symbolic{i}, \xabs{fut} \mapsto \symbolic{f}}{\xabs{f} \mapsto 0},
    \resolvev(\objname,\destiny,\methodname,\symbolic{\xabs{this.f}}_2),
    \statepair{\xabs{i} \mapsto \symbolic{i}, \xabs{fut} \mapsto \symbolic{f}}{\xabs{f} \mapsto 0}
    }\\
\end{align*}
\end{minipage}
}

\end{example}

%% file: chapters/032_global.tex
This section gives the globally concrete semantics of \coreactor. While the locally abstract semantics are a denotational semantics, the GC semantics are operational semantics, where a reduction function is used
to compute the next step in the computation. \coreactor has a two-layered GC semantics: there is a global semantics for objects and a global semantics for the whole system.
The two semantics are not independent: the global system-semantics is defined using the object-semantics and ensures that the communication steps that the object makes are correct with respect to the whole system.

\subsubsection{Semantics of Objects}
\begin{definition}[Object]
An object $\mathbf{O}$ has the form $\big(\processpool,\trace,\process\big)_\objname$, where \objname is the object name, 
\process is the \emph{active process}: a set of traces which may continue execution.
\processpool is the \emph{process pool}, a set of processes, i.e., a set of sets of traces.
One of the process may be chosen for continuation of execution once no process is active.
Finally, the trace $\trace$ is the \emph{object trace}. This is a local trace that models the execution of the object up to this point in time. 
\end{definition}

Initially, all object are initialized with an empty process pool, no active process and an object trace that has no local state and a heap initializing all fields to (semantic) default values.
The object called by the main block has a non empty process pool, with a single element: the semantics of the called method for some concrete future $f$ and an accordingly instantiated local state.
\begin{definition}[Initial Object]\label{def:iobj}
Let $\xabs{v!m(}\many{\expr}\xabs{)}$ be the initial call from the main block. The first initial object below is the initial object for all other objects, the second initial object for the called object.
\begin{align*}
&\left(\emptyset,\emptyset \triangleright \sequence{\statepair{\cdot}{\heap_i}}, \emptyset\right)_{\objname_i}\\
&\left(\left\{\eval{\xabs{m}}{\objname,\destiny,\xabs{m},\statepair{\state}{\heap_i}}\right\},\emptyset \triangleright \sequence{\statepair{\state}{\heap_i}}, \emptyset\right)_{\objname_i}
\end{align*}
Where $\heap_i$ maps all fields to the default value of their data type ({\normalfont$\mathbf{unit},0,\textbf{False},\sequence{},\mathbf{no}$}) and the parameters according to the parameters of the object creation of $\objname_i$.
The local state of the only process $\state$ is initialized according to the parameters $\many{\xabs{e}}$ of the initial call.
\end{definition}

The GC semantics of objects are the point where the locally abstract traces are executed and concretized.
This is done by \emph{agreement}: the traces in a process agree on how to continue execution by generating a local trace \emph{without symbolic values}, so the object can do one execution step.
Informally, agreement is established as follows: 
\begin{enumerate}
\item
For each trace in the process, a candidate is generated. 
The history of the candidate are the first three elements of the history of the trace.
The selection condition of the candidate is the set of all expressions in the selection condition of the trace, which have either no symbolic values or only a symbolic value introduced by the event in the history of the candidate.
The heap of the current state is applied upon this candidate trace.
If a trace cannot generate a candidate, agreement fails. This is the case if the candidate traces after application of the heap contains a value that is not introduced by its event.
\item
In the next step, agreement is established. The processes agree on a candidate, if its symbolic value can be substituted such that the candidate is fully concrete and the following holds.
For each candidate, the symbolic value is substituted by some semantic value and the selection condition is evaluated. A candidate is agreed on, if it its the (concretized) candidate of all processes whose selection condition holds.
Then, the continuation process is computed: it is the suffix of all traces whose candidate was selected and who are still selectable after the symbolic value is substituted.
\end{enumerate}
Fig.~\ref{fig:agree} gives an informal overview how a process agrees on a trace: 
Step 1 generates the candidates. The sets $\{\sexpr_1,\dots,\sexpr_{k'}\}$ are those expressions in the selection conditions, where the only symbolic value is occurring also in the candidates history. Step 2 applies the heap. Step 3 ensures that all selectable candidates are equal. Continuation is not pictured, but in this example only $\trace$ and $\trace'$ are used, $\trace''$ is discarded because its candidate is not selected.
\begin{figure}
\scalebox{0.8}{
\begin{minipage}{0.3\textwidth}
\casedr{
\{\sexpr_1,\dots,\sexpr_k\} &\triangleright \sequence{\state_1,\event,\state_2}\chop\trace \\
\{\sexpr'_1,\dots,\sexpr'_l\} &\triangleright \sequence{\state'_1,\event',\state'_2}\chop\trace' \\
\{\sexpr''_1,\dots,\sexpr''_m\} &\triangleright \sequence{\state''_1,\event'',\state''_2}\chop\trace''
}
\end{minipage}}
\hspace{1cm}$\xrightarrow{Step 1}$
\scalebox{0.8}{
\begin{minipage}{0.3\textwidth}
\casedr{
\{\sexpr_1,\dots,\sexpr_{k'}\}\heap &\triangleright \sequence{\state_1,\event,\state_2}\heap \\
\{\sexpr'_1,\dots,\sexpr'_{l'}\}\heap &\triangleright \sequence{\state'_1,\event',\state'_2}\heap \\
\{\sexpr''_1,\dots,\sexpr''_{m'}\}\heap &\triangleright \sequence{\state''_1,\event'',\state''_2}\heap
}
\end{minipage}}
\hspace{1cm}$\xrightarrow{Step 2}$
\scalebox{0.8}{
\begin{minipage}{0.3\textwidth}
\casedr{
\{\mathbf{True}\} &\triangleright \sequence{\state_1,\event,\state_2}\heap[\symbolic{v} \mapsto v] \\
\{\mathbf{True\}} &\triangleright \sequence{\state'_1,\event',\state'_2}\heap[\symbolic{v} \mapsto v] \\
\{\mathbf{False}\} &\triangleright \sequence{\state''_1,\event'',\state''_2}\heap[\symbolic{v} \mapsto v]
}
\end{minipage}}
\hspace{1cm}$\xrightarrow{Step 3} \emptyset  \triangleright \sequence{\state_1,\event,\state_2}\heap[\symbolic{v} \mapsto v]$

\caption{Structure of Agreement.}
\label{fig:agree}
\end{figure}
Before we give an example, we first introduce the required technical definitions.
\begin{definition}
Let 
\heap be a heap.
A trace $\trace = \selection \triangleright \history$ is \emph{selectable} if the natural evaluation of all expressions in the selection condition is \textbf{True}:
\[\forall \sexpr \in \selection.~\eval{\sexpr}{} = \mathbf{True}\]

The \emph{\heap-selection candidate} of a trace $\trace = \selection \triangleright \history$ is a trace $\trace_C = \selection_C \triangleright \history_C$
such that 
\begin{align*}
\history\heap &= \history_C \chop \history' \text{ for some $\history'$ and $|\history_C| = 3$.}\\
&\text{All symbolic values in $\history_C$ are introduced within.}\\
\selection_C &= \{\sexpr \in \selection\heap \mid \text{\sexpr contains at most one symbolic value, the one introduced in $\history_C$} \}
\end{align*}
\end{definition}
\begin{example}\label{ex:lasem3}
In the example~\ref{ex:lasem2}, $\trace_1, \trace_2,\trace_3$ have the following $\{\xabs{f} \mapsto 2 \}$-selection candidates:
\begin{align*}
    \emptyset \triangleright &\sequence{
    \statepair{\xabs{i} \mapsto 0, \xabs{fut} \mapsto \mathbf{no}}{\xabs{f} \mapsto 2,\xabs{o} \mapsto \objname'},
    \invocrev(\objname,\destiny,\methodname,\sequence{5}),
    \statepair{\xabs{i} \mapsto 0, \xabs{fut} \mapsto \mathbf{no}}{\xabs{f} \mapsto 2,\xabs{o} \mapsto \objname'}}\\
    \emptyset \triangleright &\sequence{
    \statepair{\xabs{i} \mapsto 0, \xabs{fut} \mapsto \mathbf{no}}{\xabs{f} \mapsto 2,\xabs{o} \mapsto \objname'},
    \invocrev(\objname,\destiny,\methodname,\sequence{5}),
    \statepair{\xabs{i} \mapsto 0, \xabs{fut} \mapsto \mathbf{no}}{\xabs{f} \mapsto 2,\xabs{o} \mapsto \objname'}}\\
    \emptyset \triangleright &\sequence{
    \statepair{\xabs{i} \mapsto 0, \xabs{fut} \mapsto \mathbf{no}}{\xabs{f} \mapsto 2,\xabs{o} \mapsto \objname'},
    \invocrev(\objname,\destiny,\methodname,\sequence{5}),
    \statepair{\xabs{i} \mapsto 0, \xabs{fut} \mapsto \mathbf{no}}{\xabs{f} \mapsto 2,\xabs{o} \mapsto \objname'}}
\end{align*}
\end{example}

Agreement generates not only the agreed upon trace and the continuation process, but also outputs the event within the candidate. 
This event is used to communicate the eventual choice for the symbolic value needed for agreement -- the global semantics will not execute a step, if the agreement relies on an event that is not possible.
E.g., if the process agrees to substitute $5$ for the read value of a future $f$, but the future was resolved with $0$.
\begin{definition}\label{def:agree}
$\process$ $\heap$-agrees on a trace $\trace_F$ with continuation $\process_F$ if the following holds.
\begin{itemize}
\item If the selection candidates have a symbolic value \symbolic{v} then there is a concrete value $v$ such that for every \heap-selection candidate $\trace$, the following holds: If $\trace[\symbolic{v} \mapsto v]$ is selectable, then it is equal to $\trace_F$.
We say that \trace is selected under $[\symbolic{v} \mapsto v]$.
\item Let $\overline{\process}$ be the set of traces, whose \heap-selection candidates are selected under their resp. $[\symbolic{v} \mapsto v]$.
\[\process_F = \{\trace_C \mid \exists \trace' \in \overline{\process}.~\trace' = ((\trace_F\megachop\trace_C)[\symbolic{v} \mapsto v])\rho\}\]
\end{itemize}
We say that $\trace_C$ is the continuation of $\trace'$ and that $\process$ $\heap$-agrees on a trace $\trace_F$ with continuation $\process_F$ under event $\trace_F[2]$.
\end{definition}

\begin{example}\label{ex:gsem1}
The above three traces $\{\xabs{f} \mapsto 2\}$-agree on the following trace under $\invocrev(\objname,\destiny,\methodname,\sequence{5})$.
\begin{align*}
\hat{\trace}  =  \emptyset \triangleright &\sequence{
    \statepair{\xabs{i} \mapsto 0, \xabs{fut} \mapsto \mathbf{no}}{\xabs{f} \mapsto 2,\xabs{o} \mapsto \objname'},
    \invocrev(\objname,\destiny,\methodname,\sequence{5}),
    \statepair{\xabs{i} \mapsto 0, \xabs{fut} \mapsto \mathbf{no}}{\xabs{f} \mapsto 2,\xabs{o} \mapsto \objname'}}
\end{align*}
The continuation is:
\\\noindent\scalebox{0.65}{
\begin{minipage}{\textwidth}
\begin{align*}
\trace_1' = \{
    2 > \symbolic{i}, \symbolic{\xabs{this.f}}_2 < 5\} \triangleright &\sequence{
    \statepair{\xabs{i} \mapsto 0, \xabs{fut} \mapsto \mathbf{no}}{\xabs{f} \mapsto 2,\xabs{o} \mapsto \objname'},
    \noev,
    \statepair{\xabs{i} \mapsto 0, \xabs{fut} \mapsto \mathbf{no}}{\xabs{f} \mapsto 3,\xabs{o} \mapsto \objname'},
    \invocev(\objname, \objname', \symbolic{f}, \xabs{n}, \sequence{})
    }\\
\circ &
\sequence{
    \statepair{\xabs{i} \mapsto 0, \xabs{fut} \mapsto \symbolic{f}}{\xabs{f} \mapsto 3,\xabs{o} \mapsto \objname'},
    \resolvrev(\objname,\symbolic{f},\symbolic{i},1),
    \statepair{\xabs{i} \mapsto \symbolic{i}, \xabs{fut} \mapsto \symbolic{f}}{\xabs{f} \mapsto 3,\xabs{o} \mapsto \objname'},
    \suspev,
    \statepair{\xabs{i} \mapsto \symbolic{i}, \xabs{fut} \mapsto \symbolic{f}}{\xabs{f} \mapsto 3,\xabs{o} \mapsto \objname'},
    \marker
    }\\
\circ &
\sequence{
    \statepair{\xabs{i} \mapsto \symbolic{i}, \xabs{fut} \mapsto \symbolic{f}}{\xabs{f} \mapsto \symbolic{\xabs{this.f}}_2,\xabs{o} \mapsto \objname'},
    \susprev,
    \statepair{\xabs{i} \mapsto \symbolic{i}, \xabs{fut} \mapsto \symbolic{f}}{\xabs{f} \mapsto \symbolic{\xabs{this.f}}_2,\xabs{o} \mapsto \objname'},
    \noev,
    \statepair{\xabs{i} \mapsto \symbolic{i}, \xabs{fut} \mapsto \symbolic{f}}{\xabs{f} \mapsto \symbolic{0},\xabs{o} \mapsto \objname'}}\\
    \circ&\sequence{\resolvev(\objname,\destiny,\methodname,\symbolic{\xabs{this.f}}_2,
    \statepair{\xabs{i} \mapsto \symbolic{i}, \xabs{fut} \mapsto \symbolic{f}}{\xabs{f} \mapsto \symbolic{0},\xabs{o} \mapsto \objname'}
    }\\[5mm]
\trace_2' = \{2 > \symbolic{i}, \symbolic{\xabs{this.f}}_2 \geq 5\} \triangleright &\sequence{
    \statepair{\xabs{i} \mapsto 0, \xabs{fut} \mapsto \mathbf{no}}{\xabs{f} \mapsto 2,\xabs{o} \mapsto \objname'},
    \noev,
    \statepair{\xabs{i} \mapsto 0, \xabs{fut} \mapsto \mathbf{no}}{\xabs{f} \mapsto 3,\xabs{o} \mapsto \objname'},
    \invocev(\objname, \objname', \symbolic{f}, \xabs{n}, \sequence{})
    }\\
\circ &
\sequence{
    \statepair{\xabs{i} \mapsto 0, \xabs{fut} \mapsto \symbolic{f}}{\xabs{f} \mapsto 3,\xabs{o} \mapsto \objname'},
    \resolvrev(\objname,\symbolic{f},\symbolic{i},1),
    \statepair{\xabs{i} \mapsto \symbolic{i}, \xabs{fut} \mapsto \symbolic{f}}{\xabs{f} \mapsto 3,\xabs{o} \mapsto \objname'}
    \suspev,
    \statepair{\xabs{i} \mapsto \symbolic{i}, \xabs{fut} \mapsto \symbolic{f}}{\xabs{f} \mapsto 3,\xabs{o} \mapsto \objname'},
    \marker
    }\\
\circ &
\sequence{
    \statepair{\xabs{i} \mapsto \symbolic{i}, \xabs{fut} \mapsto \symbolic{f}}{\xabs{f} \mapsto \symbolic{\xabs{this.f}}_2,\xabs{o} \mapsto \objname'},
    \susprev,
    \statepair{\xabs{i} \mapsto \symbolic{i}, \xabs{fut} \mapsto \symbolic{f}}{\xabs{f} \mapsto \symbolic{\xabs{this.f}}_2,\xabs{o} \mapsto \objname'}}\\
    \circ & \sequence{  \resolvev(\objname,\destiny,\methodname,\symbolic{\xabs{this.f}}_2),
    \statepair{\xabs{i} \mapsto \symbolic{i}, \xabs{fut} \mapsto \symbolic{f}}{\xabs{f} \mapsto \symbolic{\xabs{this.f}}_2,\xabs{o} \mapsto \objname'},
    }\\[5mm]
\trace_3' = \{2 \leq \symbolic{i}\} \triangleright &\sequence{
    \statepair{\xabs{i} \mapsto 0, \xabs{fut} \mapsto \mathbf{no}}{\xabs{f} \mapsto 2,\xabs{o} \mapsto \objname'},
    \noev,
    \statepair{\xabs{i} \mapsto 0, \xabs{fut} \mapsto \mathbf{no}}{\xabs{f} \mapsto 3,\xabs{o} \mapsto \objname'},
    \invocev(\objname, \objname', \symbolic{f}, \xabs{n}, \sequence{}),
    }\\
\circ &
\sequence{
    \statepair{\xabs{i} \mapsto 0, \xabs{fut} \mapsto \symbolic{f}}{\xabs{f} \mapsto 3,\xabs{o} \mapsto \objname'},
    \resolvrev(\objname,\symbolic{f},\symbolic{i},1),
    \statepair{\xabs{i} \mapsto \symbolic{i}, \xabs{fut} \mapsto \symbolic{f}}{\xabs{f} \mapsto 3,\xabs{o} \mapsto \objname'},
    \resolvev(\objname,\destiny,\methodname,3),
    \statepair{\xabs{i} \mapsto \symbolic{i}, \xabs{fut} \mapsto \symbolic{f}}{\xabs{f} \mapsto 3,\xabs{o} \mapsto \objname'},
    }
\end{align*}
\end{minipage}}

\noindent
Next, the continuation agrees on the following (for any heap):
\begin{align*}
\hat{\trace}' = \emptyset \triangleright &\sequence{
    \statepair{\xabs{i} \mapsto 0, \xabs{fut} \mapsto \mathbf{no}}{\xabs{f} \mapsto 2,\xabs{o} \mapsto \objname'},
    \noev,
    \statepair{\xabs{i} \mapsto 0, \xabs{fut} \mapsto \mathbf{no}}{\xabs{f} \mapsto 3,\xabs{o} \mapsto \objname'}}
\end{align*}
An then on the following, again for any heap and any future $f$:
\begin{align*}
\hat{\trace}'' = \emptyset \triangleright &\sequence{
    \statepair{\xabs{i} \mapsto 0, \xabs{fut} \mapsto \mathbf{no}}{\xabs{f} \mapsto 3,\xabs{o} \mapsto \objname'},
    \invocev(\objname, \objname', f, \xabs{n}, \sequence{}),
    \statepair{\xabs{i} \mapsto 0, \xabs{fut} \mapsto f}{\xabs{f} \mapsto 3,\xabs{o} \mapsto \objname'}}
\end{align*}
Note that they can agree on any future -- the global system semantics later forces an agreement on a fresh future.
Finally, there is a split, they can agree on, among others, one of the following two traces, depending on the value that is instantiated for $\symbolic{i}$.
\begin{align*}
\hat{\trace}''' = \{2 > 0\} \triangleright &\sequence{
    \statepair{\xabs{i} \mapsto 0, \xabs{fut} \mapsto f}{\xabs{f} \mapsto 3,\xabs{o} \mapsto \objname'},
    \resolvrev(\objname,f,0,1),
    \statepair{\xabs{i} \mapsto 0, \xabs{fut} \mapsto f}{\xabs{f} \mapsto 3,\xabs{o} \mapsto \objname'}}\\
\hat{\trace}'''' = \{2 \leq 5\} \triangleright &\sequence{
    \statepair{\xabs{i} \mapsto 5, \xabs{fut} \mapsto f}{\xabs{f} \mapsto 3,\xabs{o} \mapsto \objname'},
    \resolvrev(\objname,f,5,1),
    \statepair{\xabs{i} \mapsto 5, \xabs{fut} \mapsto f}{\xabs{f} \mapsto 3,\xabs{o} \mapsto \objname'}}
\end{align*}
The difference now is that the first trace selects the first branch of the branching (and its continuation set has two elements), while the second trace selects the empty \abs{else} branch and has only one element in its continuation set. Again, the global system semantics ensures later that the value the process agrees on to substituted as the read value is indeed the value the future was resolved with.
\end{example}

\begin{definition}[Object Semantics]
The semantics of an object is a relation $\xrightarrow{\event}$ between objects, where \event is either an event or a marked event $\overline{\event}$. 
We say that \event is communicated to the outside.
The semantics consists of four rules, given in Fig.~\ref{fig:gcosem}.
\begin{itemize}
\item
Rule \rulename{O-Schedule} activates a process in the process pool, if its traces agree on a trace and no other process is active.
This trace is executed by being appended to the object trace and making the continuation to the active process.
The event under which the traces agree is communicated to the outside. 
It is not possible to active a process that immediately suspends -- note that if two \abs{await} statement follow each other directly, then there is still a suspension reaction event between their suspension events.
\item 
Rule \rulename{O-Step} is analogous to \rulename{O-Schedule}, but the process is already active.
\item
Rule \rulename {O-Deschedule} is analogous to \rulename{O-Step}, but if the communicated event is suspending, then the continuation is added to the process pool and no active process exists.
\item
Rule \rulename{O-Add} finally adds a called method to the process pool. This rule can always be executed, the system semantics ensure that every called method is added exactly once.
\end{itemize}
Only the \rulename{O-Schedule} rules uses the $\megachop$ operator, because this is the only point where the local state of the next step may not agree with the local state of the last step.
There is no rule for process termination. If a process terminates, then the continuation is the empty set and this falls under the rule \rulename{O-Step}.
\begin{figure}
\SSINFER{O-Schedule}{
     \text{$\event$ is neither $\awaitev$ nor $\suspev$}
}{
\mathsf{last}(\trace) = \defstatepair \qquad \text{$\Theta$ $\heap$-agrees on $\trace'$ with continuation $\Theta'$ under $\event$}
}{
\big((\traceset_i)_{i\in I}\cup\{\Theta\},\trace,\emptyset\big)_\objname \xrightarrow{\event} \big((\traceset_i)_{i\in I},\trace\megachop\trace',\Theta'\big)_\objname
}
\SSINFER{O-Step}{
     \text{$\event$ is neither $\awaitev$ nor $\suspev$}\qquad
}{
     \mathsf{last}(\trace) = \defstatepair \qquad \text{$\Theta$ $\heap$-agrees on $\trace'$ with continuation $\Theta'$ under $\event$}
}{
    \big((\traceset_i)_{i\in I},\trace,\Theta\big)_\objname \xrightarrow{\event} \big((\traceset_i)_{i\in I},\trace\chop\trace',\Theta'\big)_\objname
}
\SSINFER{O-Deschedule}{
     \text{$\event$ is either $\awaitev$ or $\suspev$}
}{
     \mathsf{last}(\trace) = \defstatepair \qquad \text{$\Theta$ $\heap$-agrees on $\trace'$ with continuation $\Theta'$ under $\event$}
}{
    \big((\traceset_i)_{i\in I},\trace,\Theta\big)_\objname \xrightarrow{\event} \big((\traceset_i)_{i\in I}\cup\{\Theta'\},\trace\chop\trace',\emptyset\big)_\objname
}
\SINFER{O-Add}{
}{
\big((\traceset_i)_{i\in I},\trace,\Theta\big)_\objname    \xrightarrow{\overline{\invocev(\objname',\objname,f,\methodname,\many{\sexpr})}}
\big((\traceset_i)_{i\in I}\cup\{\eval{\xabs{m}}{\objname,\destiny,\xabs{m},\statepair{\state}{\heapid}}\},\trace,\Theta\big)_\objname
}
\caption{Globally Concrete Semantics of Objects.}
\label{fig:gcosem}
\end{figure}
\end{definition}

\begin{example}
The above example~\ref{ex:gsem1} can be formulated as follows with the object semantics. For readability's sake we omitted the continuations and the last trace before suspension.
\begingroup
\allowdisplaybreaks
\begin{align*}
&\left(\emptyset,\statepair{\cdot}{\xabs{f} \mapsto 2},\emptyset\right)_\objname \\
\xrightarrow{\overline{\invocev(\objname,\methodname,\destiny,\sequence{5})}}
&\left(\{\{\trace_1,\trace_2,\trace_3\}\},\statepair{\cdot}{\xabs{f} \mapsto 2},\emptyset\right)_\objname \\
\xrightarrow{\invocrev(\objname,\destiny,\methodname,\sequence{5})}
&\left(\emptyset,\statepair{\cdot}{\xabs{f} \mapsto 2}\megachop\hat{\trace},\{\trace_1',\trace_2',\trace_3'\}\right)_\objname \\
\xrightarrow{\noev}
&\left(\emptyset,\statepair{\cdot}{\xabs{f} \mapsto 2}\megachop\hat{\trace}\chop\hat{\trace}',\{...\}\right)_\objname\\
\xrightarrow{\invocev(\objname, \objname', f, \xabs{n}, \sequence{})}
&\left(\emptyset,\statepair{\cdot}{\xabs{f} \mapsto 2}\megachop\hat{\trace}\chop\hat{\trace}'\chop\hat{\trace}'',\{...\}\right)_\objname\\
\xrightarrow{\resolvrev(\objname,f,0,1)}
&\left(\emptyset,\statepair{\cdot}{\xabs{f} \mapsto 2}\megachop\hat{\trace}\chop\hat{\trace}'\chop\hat{\trace}''\chop\hat{\trace}''',\{...\}\right)_\objname\\
\xrightarrow{\suspev}
&\left(\{...\},\statepair{\cdot}{\xabs{f} \mapsto 2}\megachop\hat{\trace}\chop\hat{\trace}'\chop\hat{\trace}''\chop\hat{\trace}'''\chop...,\emptyset\right)_\objname
\end{align*}
\endgroup
\end{example}

\subsubsection{Semantics of Programs}
Objects are always reduced in the context of a system, that enforces that the steps a process agrees on to execute next adheres to the values already used in the system.
\begin{definition}[Systems]
A \emph{system} \textbf{S} has the form $\mathsf{evs}\blacktriangleright\mathsf{obs}$, where $\mathsf{evs}$ is a sequence of events and $\mathsf{obs}$ consists of objects:
\[\mathsf{obs} ::= \mathbf{O} \sep \mathsf{obs}~\mathsf{obs}\]
The composition $\mathsf{obs}~\mathsf{obs}$ is commutative and associative: $\mathsf{obs}~\mathsf{obs}' = \mathsf{obs}'~\mathsf{obs}$ and $\mathsf{obs}~(\mathsf{obs}'~\mathsf{obs}'') = (\mathsf{obs}~\mathsf{obs}')~\mathsf{obs}''$.
\end{definition}

\begin{definition}[Initial, Terminated and Stuck Systems]
A system is terminated if all its objects have (1) an empty process pool and (2) no active process.
A system is stuck if it cannot be reduced further, but is not terminated.
Given a main block
\\\noindent\begin{abscode}
main{ 
    C$_1$ v$_1$ = new C$_1$(e);
    ...
    C$_n$ v$_n$ = new C$_n$(e);
    v$_j$!m();
}
\end{abscode}

The initial system has the form 
\[\sequence{}\blacktriangleright\mathbf{O}_1~\dots~\mathbf{O}_n\]
Where the initial objects are as defined in definition.~\ref{def:iobj}
\end{definition}

\begin{definition}[System Semantics]
Fig.~\ref{fig:gcssem} shows the rules for system semantics.
\begin{itemize}
\item \rulename{S-Internal} just applies one rule for some object and records the event in the global history.
\item \rulename{S-Get} applies a rule for some object that communicates a resolving reaction event. It is checked that the read value that the process agreed on is the value that the read future is resolved with. 
\item Rule \rulename{S-Invoc} applies a rule for some object that communicates an invocation event. It adds the called method by applying \rulename{O-Add} on the called object. 
Only the communicated invocation event is added to the global history.
\end{itemize}
\begin{figure}
\SINFER{S-Internal}{
\mathbf{O} \xrightarrow{\event} \mathbf{O}'\qquad\event\text{ is neither a \invocev nor a \resolvrev nor has the form $\overline{\event}$}
}{
 \mathsf{evs}\blacktriangleright\mathbf{O}~\mathsf{obs} \xrightarrow{\event} \mathsf{evs} \circ \sequence{\event}\blacktriangleright\mathbf{O}'~\mathsf{obs}
}
\SINFER{S-Get}{
\mathbf{O} \xrightarrow{\event} \mathbf{O}'\qquad\event=\resolvrev(\objname,f,\xabs{v},i) \qquad \resolvev(\objname',f,\methodname,\xabs{v}) \text{ is in $\mathsf{evs}$ for some $\objname',\methodname$}
}{
 \mathsf{evs}\blacktriangleright\mathbf{O}~\mathsf{obs} \xrightarrow{\event} \mathsf{evs} \circ \sequence{\event}\blacktriangleright\mathbf{O}'~\mathsf{obs}
}
\SINFER{S-Invoc}{
\mathbf{O}_1 \xrightarrow{\invocev(\objname,\objname',f,\methodname,\many{\xabs{e}})} \mathbf{O}_1' \qquad  \mathbf{O}_2 \xrightarrow{\overline{\invocev(\objname,\objname',f,\methodname,\many{\xabs{e}})}} \mathbf{O}_2' \qquad \text{$f$ is fresh in $\mathsf{evs}$}
}{
 \mathsf{evs}\blacktriangleright\mathbf{O}_1~\mathbf{O}_2~\mathsf{obs} \xrightarrow{\invocev(\objname,\objname',f,\methodname,\many{\xabs{e}})} \mathsf{evs} \circ \sequence{\invocev(\objname,\objname',f,\methodname,\many{\xabs{e}})}\blacktriangleright\mathbf{O}_1'~\mathbf{O}_2'~\mathsf{obs}
}
\caption{Globally Concrete Semantics of Systems.}
\label{fig:gcssem}
\end{figure}
\end{definition}

Analogously to local traces, we define global traces.
A global trace is a sequence of pairs, where each pair is the global state at this point of execution, and the last executed event.
A global state only records the current heaps of all objects, not the traces, processes and process pools.
\begin{definition}[Global Traces and Big-Step Semantics]
A global state $\mathsf{gl}$ is a function from object names to heaps. Given a system 
\\\noindent\scalebox{0.9}{\begin{minipage}{\textwidth}
\[\mathbf{S} = \mathsf{evs}~\left(\processpool_1,\trace_1\megachop\sequence{\statepair{\state_1}{\heap_1}},\process_1\right)_{\objname_1}~\dots~\left(\processpool_n,\trace_n\megachop\sequence{\statepair{\state_n}{\heap_n}},\process_n\right)_{\objname_n}\]
\end{minipage}}

The corresponding global state $\mathsf{gl}(\mathbf{S})$ is $\{\objname_1 \mapsto \heap_1,\dots,\objname_n \mapsto \heap_n\}$.

A global trace $\globaltrace$ is a sequence global states and events, analogous to local histories.
Given a run of a system 
\[\mathbf{S}_1\xrightarrow{\event_2}\mathbf{S}_2\dots\xrightarrow{\event_n}\mathbf{S}_n\]
its global trace is 
\[\sequence{\mathsf{gl}(\mathbf{S}_1),\event_2,\mathsf{gl}(\mathbf{S}_2),\dots,\event_n, \mathsf{gl}(\mathbf{S}_n)}\]
Given a program \prgm, we say that \prgm realizes a global trace $\globaltrace$, if there is a run from its initial state to a terminated state, such that $\globaltrace$ is the trace of that run. We write this as
\[\prgm\Downarrow\globaltrace\]
\end{definition}
Symbolic values do not escape locally abstract semantics, neither the object trace nor the global trace contains symbolic values.
\begin{lemma}
Let \prgm be a program and $\globaltrace$ a global trace with $\prgm\Downarrow\globaltrace$. $\globaltrace$ contains no symbolic values.
\end{lemma}

%% file: chapters/035_select.tex
The local semantics of a method can be expressed as $\eval{\xabs{m}}{\symbolic{\objname},\symbolic{\destiny},\xabs{m},\statepair{\symbolic{\state}}{\symbolic{\heap}}}$, where $\symbolic{\state}$ and $\symbolic{\heap} = \heapid^i(\heap)$ are fully symbolic.
The heap $\symbolic{\state}$ also maps all parameters to symbolic object names.
This set describes the behavior of a method in every possible context, but also vastly overapproximate it, especially when some context is known.
Such overapproximation manifests in several ways: 
\begin{itemize}
\item Some symbolic traces are never selected, because their selection condition can never hold.
\item Some symbolic traces are not selected in some context, because in this context their selection condition can never hold.
\item Some concrete instantiantions of symbolic traces are not possible in some or every context.
\end{itemize}
\begin{example}\label{ex:select}
Consider the method in Fig.~\ref{fig:secex}, which is in some class with a field \abs{Int i = 10}:
\begin{figure}
\begin{abscode}
Int m(Int j){
    if( j> 0 ) {
        if( j> 0 ){ this.i = this.i + 1 } 
        else      { this.i = -1 }
    } else {
        if(this.i >= 0){ this.i = this.i + 1 }
        else { this.i = -1 }
    }
    return this.i;
}
\end{abscode}
\caption{Example Method for Selectability}
\label{fig:secex}
\end{figure}
The semantics contains four trace, one for each branch (with simplified events), which are shown in Fig.~\ref{fig:secex2}.
\begin{figure}
\scalebox{0.8}{\begin{minipage}{\textwidth}
\begin{align}
\{\symbolic{j} > 0, \symbolic{j} > 0\} 
&\triangleright \sequence{
    \statepair{\{\xabs{j} \mapsto \symbolic{j}\}}{\{\xabs{i} \mapsto \symbolic{i}\}},
    \invocrev,
    \statepair{\{\xabs{j} \mapsto \symbolic{j}\}}{\{\xabs{i} \mapsto \symbolic{i}\}},
    \noev,
    \statepair{\{\xabs{j} \mapsto \symbolic{j}\}}{\{\xabs{i} \mapsto \symbolic{i}+1\}},
    \resolvev,
    \statepair{\{\xabs{j} \mapsto \symbolic{j}\}}{\{\xabs{i} \mapsto \symbolic{i}+1\}}
}\label{ex:sel:1}\\
\{\symbolic{j} > 0, \symbolic{j} <= 0\}
&\triangleright \sequence{
    \statepair{\{\xabs{j} \mapsto \symbolic{j}\}}{\{\xabs{i} \mapsto \symbolic{i}\}},
    \invocrev,
    \statepair{\{\xabs{j} \mapsto \symbolic{j}\}}{\{\xabs{i} \mapsto \symbolic{i}\}},
    \noev,
    \statepair{\{\xabs{j} \mapsto \symbolic{j}\}}{\{\xabs{i} \mapsto -1\}},
    \resolvev,
    \statepair{\{\xabs{j} \mapsto \symbolic{j}\}}{\{\xabs{i} \mapsto -1\}}
}\label{ex:sel:2}\\
\{\symbolic{j} <= 0, \symbolic{i} >= 0\}
&\triangleright \sequence{
    \statepair{\{\xabs{j} \mapsto \symbolic{j}\}}{\{\xabs{i} \mapsto \symbolic{i}\}},
    \invocrev,
    \statepair{\{\xabs{j} \mapsto \symbolic{j}\}}{\{\xabs{i} \mapsto \symbolic{i}\}},
    \noev,
    \statepair{\{\xabs{j} \mapsto \symbolic{j}\}}{\{\xabs{i} \mapsto \symbolic{i}+1\}},
    \resolvev,
    \statepair{\{\xabs{j} \mapsto \symbolic{j}\}}{\{\xabs{i} \mapsto \symbolic{i}+1\}}
}\label{ex:sel:3}\\
\{\symbolic{j} <= 0, \symbolic{i} < 0\}
&\triangleright \sequence{
    \statepair{\{\xabs{j} \mapsto \symbolic{j}\}}{\{\xabs{i} \mapsto \symbolic{i}\}},
    \invocrev,
    \statepair{\{\xabs{j} \mapsto \symbolic{j}\}}{\{\xabs{i} \mapsto \symbolic{i}\}},
    \noev,
    \statepair{\{\xabs{j} \mapsto \symbolic{j}\}}{\{\xabs{i} \mapsto -1\}},
    \resolvev,
    \statepair{\{\xabs{j} \mapsto \symbolic{j}\}}{\{\xabs{i} \mapsto -1\}}
}\label{ex:sel:4}
\end{align}
\end{minipage}}
\caption{Traces of the method of Fig.~\ref{fig:secex}.}
\label{fig:secex2}
\end{figure}
Obviously, the trace \ref{ex:sel:2} is never selected, no matter how the rest of the program looks like, as no semantic value can be substituted for $\symbolic{j}$ such that both conditions are true.

With some context being known, we can exclude further traces.
If there is no further method in this class, than trace \ref{ex:sel:4} is also never selected:
\abs{this.i} is not negative in the initial state and so the forth trace is not selected in the first run of the method -- but the field stays non-negative and so every further run also never selects this trace.
Finally, trace \ref{ex:sel:1} is an abstraction for the following concrete trace:

\noindent\scalebox{0.9}{\begin{minipage}{\textwidth}
\[
\{\textbf{True},\textbf{True}\} 
\triangleright \sequence{
    \statepair{\{\xabs{j} \mapsto 1\}}{\{\xabs{i} \mapsto 5\}},
    \invocrev,
    \statepair{\{\xabs{j} \mapsto 1\}}{\{\xabs{i} \mapsto 5\}},
    \noev,
    \statepair{\{\xabs{j} \mapsto 1\}}{\{\xabs{i} \mapsto 6\}},
    \resolvev,
    \statepair{\{\xabs{j} \mapsto 1\}}{\{\xabs{i} \mapsto 6\}}
}\]
\end{minipage}}

But, for the same reason as why trace \ref{ex:sel:4} is not selected, the symbolic value $\symbolic{i}$ is never substituted with $5$. Indeed, none of the traces substitutes $\symbolic{i}$ with a value below $10$.

We stress that the above reasoning only requires the class to be known, not the whole program.
\end{example}

For precision it is important to ensure that formal verification only considers traces that really correspond to a possible run, otherwise verification may fail with a false negative: verification may fails, but only for some behavior that does not correspond to a possible run.
This is akin to the restriction on reachable states in, e.g., model checking, instead of checking all possible states of a system.
While it is not possible to exclude all traces that do not corresponds to real run, one can still aim to be as precise as possible.
Before we can formalize the reasoning about the above example, we require some technical tools to be more precise about substitution.
\begin{definition}[Initial Continuation]
Let $\symbolic{\state}$ and $\symbolic{\heap}$ be fully symbolic and let $\trace \in \eval{\xabs{m}}{\symbolic{\objname},\symbolic{\destiny},\xabs{m},\statepair{\symbolic{\state}}{\symbolic{\heap}}}$ be a trace of \abs{m}.
Let there is some application of \rulename{O-Add} with $\trace' \in \eval{\xabs{m}}{\objname,\destiny,\xabs{m},\statepair{\state}{\heapid}}$ such that there are 
symbolic values $\{\xabs{v}_i \mapsto \symbolic{v}_i\}_{i\in I} \in \symbolic{\state}$ and $\{\xabs{v}_i \mapsto v_i\}_{i \in I} \in \state$ such that 
\[\trace[\symbolic{v}_i \mapsto v_i] = \trace'\]
We say that $\trace'$ is an initial continuation of $\trace$. 

In the following, whenever we refer to continuations, we mean the transitive closure formed by Definition~\ref{def:agree} and this.
\end{definition}
\begin{definition}[Concretizers]
A concretizer $\concr$ is a function from symbolic fields and symbolic values to semantics values.
\end{definition}
Note that neither local states nor heaps are concretizers.
\begin{definition}[Used Concretizers]
Let $\symbolic{\state}$ and $\symbolic{\heap}$ be fully symbolic and let $\trace \in \eval{\xabs{m}}{\symbolic{\objname},\symbolic{\destiny},\xabs{m},\statepair{\symbolic{\state}}{\symbolic{\heap}}}$ be a trace of \abs{m}.
A used concretizer $\concr$ for trace in a run of \prgm is constructed as follows.
\begin{itemize}
\item There is some application of \rulename{O-Add} with $\trace' \in \eval{\xabs{m}}{\objname,\destiny,\xabs{m},\statepair{\state}{\heapid}}$ such that there are 
symbolic values $\{\xabs{v}_i \mapsto \symbolic{v}_i\}_{i\in I} \in \symbolic{\state}$ and $\{\xabs{v}_i \mapsto v_i\}_{i \in I} \in \state$ such that 
\[\trace[\symbolic{v}_i \mapsto v_i]_{i \in I} = \trace'\]
Then we set
\[\concr(\symbolic{v})_i = v_i\text{ for all $i \in I$}\]
\item For each agreement where $\trace'$ (or one of its continuations) agrees under some substitution $[\symbolic{v} \mapsto v]$. 
\[\concr(\symbolic{v}) = v\]
\item For each agreement where $\trace'$ (or one of its continuations) agrees under some heap which substitutes $\symbolic{\xabs{this.f}}_i$ by a concrete value $v$
\[\concr(\symbolic{\xabs{this.f}}_i) = v\]
\end{itemize}
We write the application of a concretizer to a trace $\trace$ (history, etc.) as $\trace\concr$
\end{definition}

\begin{example}
Example~\ref{ex:select} shows fully symbolic traces.
In the example~\ref{ex:lasem1}, one fully symbolic trace is the following:
\\\noindent\scalebox{0.825}{
    \begin{minipage}{\textwidth}
\begin{align*}
\trace_{\mathsf{symb}} &= \{
     \symbolic{\xabs{this.f}}_1 > \symbolic{i}, \symbolic{\xabs{this.f}}_2 < \symbolic{\xabs{p}}\} \triangleright \\&\sequence{
    \statepair{\xabs{i} \mapsto 0, \xabs{fut} \mapsto \mathbf{never}}{\xabs{f} \mapsto  \symbolic{\xabs{this.f}}_1,\xabs{o} \mapsto \symbolic{\objname'}},
    \noev,
    \statepair{\xabs{i} \mapsto 0, \xabs{fut} \mapsto \mathbf{never}}{\xabs{f} \mapsto \symbolic{\xabs{this.f}}_1,\xabs{o} \mapsto \symbolic{\objname'}},
    \invocev(\objname, \symbolic{\objname'}, \symbolic{f}, \xabs{n}, \sequence{})
    }\\
\circ &
\sequence{
    \statepair{\xabs{i} \mapsto 0, \xabs{fut} \mapsto \symbolic{f}}{\xabs{f} \mapsto \symbolic{\xabs{this.f}}_1+1,\xabs{o} \mapsto \symbolic{\objname'}},
    \resolvrev(\objname,\symbolic{f},\symbolic{i},1),
    \statepair{\xabs{i} \mapsto \symbolic{i}, \xabs{fut} \mapsto \symbolic{f}}{\xabs{f} \mapsto \symbolic{\xabs{this.f}}_1+1,\xabs{o} \mapsto \symbolic{\objname'}},
    \suspev}\\
\circ&\sequence{
    \statepair{\xabs{i} \mapsto \symbolic{i}, \xabs{fut} \mapsto \symbolic{f}}{\xabs{f} \mapsto \symbolic{\xabs{this.f}}_1+1,\xabs{o} \mapsto \symbolic{\objname'}},
    \marker
    \statepair{\xabs{i} \mapsto \symbolic{i}, \xabs{fut} \mapsto \symbolic{f}}{\xabs{f} \mapsto \symbolic{\xabs{this.f}}_2,\xabs{o} \mapsto \symbolic{\objname'}},
    \susprev
    }\\
\circ &
\sequence{
    \statepair{\xabs{i} \mapsto \symbolic{i}, \xabs{fut} \mapsto \symbolic{f}}{\xabs{f} \mapsto \symbolic{\xabs{this.f}}_2,\xabs{o} \mapsto \symbolic{\objname'}},
    \noev,
    \statepair{\xabs{i} \mapsto \symbolic{i}, \xabs{fut} \mapsto \symbolic{f}}{\xabs{f} \mapsto \symbolic{0},\xabs{o} \mapsto \symbolic{\objname'}},
    \resolvev(\objname,\destiny,\methodname,\symbolic{\xabs{this.f}}_2),
    \statepair{\xabs{i} \mapsto \symbolic{i}, \xabs{fut} \mapsto \symbolic{f}}{\xabs{f} \mapsto \symbolic{0},\xabs{o} \mapsto \symbolic{\objname'}}
    }
\end{align*}
\end{minipage}}

A used concretizer for $\trace_{\mathsf{symb}}$ is, as outlined by $\trace_1$ and its continuations in examples~\ref{ex:lasem1} to ~\ref{ex:gsem1} is the following (if during the suspension the heap does not change):
\[\concr_1 = \{\symbolic{\objname'} \mapsto \objname, \symbolic{\xabs{p}} \mapsto 5, \symbolic{\xabs{this.f}}_1 \mapsto 2, \symbolic{f} \mapsto f, \symbolic{i} \mapsto 0, \symbolic{\xabs{this.f}}_2 \mapsto 3\}\]
\end{example}

\begin{definition}
A trace $\trace$ is never selected in \prgm if it has no used concretizers.
A trace $\trace$ is selected as $\trace\concr$ in \prgm if $\concr$ is a used concretizer of $\trace$ by some terminating run of $\prgm$. 
\end{definition}

We stress that this is only defined for traces of methods in a given program, not arbitrary traces.
To include continuations, we require the notion of sharing, because the continuations generated by agreement are not suffixes, but partially apply the concretizer.

\begin{definition}
A trace $\trace = \selection \triangleright \history $ shares a concretizer $\concr$ with a trace $\trace'  = \selection' \triangleright \history'$. 
If there is a sub-concretizer $\concr' \subseteq \concr$ such that $\trace\concr' = \trace'' \megachop \trace'$ for some $\trace''$.
\end{definition}
I.e.,  a trace $\trace$ shares a concretizer with another trace $\trace'$, if $\trace'$ is the suffix of the trace resulting from applying a part of $\concr$ on $\trace$.

\begin{example}
In the above example, $\trace_{\mathsf{symb}}$ shares $\concr_1$ with $\trace_1$, where the sub-concretizer is 
\[\{\symbolic{\xabs{p}} \mapsto 5,\symbolic{\objname'} \mapsto \objname\}\]
$\trace_1$ shares $\concr_1$ with $\trace_1'$ and sub-concretizer is 
\[\{\symbolic{\xabs{this.f}}_1 \mapsto 2\}\]
\noindent Sharing is transitive:  $\trace_{\mathsf{symb}}$ shares $\concr_1$ with $\trace_1$, where the sub-concretizer is 
\[\{\symbolic{\objname'} \mapsto \objname, \symbolic{\xabs{p}} \mapsto 5, \symbolic{\xabs{this.f}}_1 \mapsto 2\}\]
\end{example}
Concretizers and continuations allow us to restrict us to used traces for reasoning. We can reject never selected traces, as illustrated in the above example.

Given a trace stemming from the semantics of a method, sharing allows us to reason about intermediate traces.
Given a subtrace we may now also restrict the possible values the symbolic values can take, by arguing that the subtrace is a continuation of some trace from the semantics of a method and thus shares its concretizers. 

\begin{lemma}
Let $\trace$ be a trace of some method \methodname. It shares concretizers with all its continuations.
\end{lemma}

We can now also formalize our statements in example~\ref{ex:select}.
\begin{example}
If there is no further method in the class of \methodname, then there is no used concretizer $\concr$ for any trace in Ex.~\ref{ex:select} that has $\concr(\symbolic{i}) = 5$
The traces \ref{ex:sel:2} and \ref{ex:sel:4} are never selected.
\end{example}

As discussed, we are only interested in selected traces and their used concretizers. Of course, this is in general undecidable, but we will safely overapproximate the set of used concretizers in the following analyses.
We used the following notation for all relevant traces of some statement.
\begin{definition}[Selected Traces]
Let \statement be a statement within a method \methodname in a program \prgm, \objname, \future be symbolic object name and future and \defstatepair a (possible partially symbolic) object state.
The set of selected traces of \statement, written $\defeval{\statement}^\prgm$ is defined as the set of continuations of \methodname that start with \statement and \defstatepair from selected traces in the local semantics of \methodname.
\end{definition}

%% file: chapters/033_logic.tex
We use two logics to \emph{express} properties of states and traces. 
The state logic is a standard first-order logic.
The trace logic is a monadic second-order logic, that embeds the state logic by using them similar to predicates on states. Similarly, they have event terms that allow to specify events.
\subsubsection{Local First-Order State Logic (FOS)}
\begin{definition}[Syntax]\label{def:FOS}
Let $p$ range over predicate symbols, $f$ over function symbols, $x$ over logical variable names and $S$ over sorts. 
As sorts we take all data types $\dType$, all class names and additionally $\mathbb{N}$ and $\mathbf{Heap}$. The logical heaps are functions from field names to semantic values.
Formulas $\phi$ and terms $\mathsf{t}$ are defined by the following grammar, where \abs{v} are program variables, consisting of local variables and the special variable \abs{heap}, and \abs{f} are all field names.
\begin{align*}
\phi &::= p(\many{\mathsf{t}}) \sep \mathsf{t} \doteq \mathsf{t} \sep \phi \vee \phi \sep \neg \phi \sep \exists x\in S.~\phi\qquad\mathsf{t} ::= x \sep \xabs{v} \sep \xabs{f} \sep f(\many{\mathsf{t}}) 
\end{align*}
\end{definition}

We demand the usual constants and that each operator defined in syntactic expressions \expr is a function symbol, so one can directly translate a syntactic expression into a FOS term.
We, thus, assume the existence of the following function symbols (where we write the binary function symbols in infix-form):
\begin{align*}
\xabs{Never} &\sep \xabs{Nil} \sep \xabs{True} \sep \xabs{False} \sep \xabs{len}(\mathsf{t}) \sep \xabs{hd}(\mathsf{t}) \sep \xabs{tl}(\mathsf{t}) \sep \xabs{Cons}(\mathsf{t},\mathsf{t}) \sep \mathsf{t} \sim \mathsf{t} \sep \xabs{!}\mathsf{t} \sep \xabs{-}\mathsf{t}
\sep \mathsf{t}[\mathsf{t}] \\&\sep \mathsf{select}(\mathsf{t},\mathsf{t}) \sep \mathsf{store}(\mathsf{t},\mathsf{t},\mathsf{t}) \sep \mathsf{anon}(\mathsf{t})
\end{align*}
Where $\mathsf{t}[\mathsf{t}]$ is indexed list access, $\mathsf{select}(h,\xabs{f})$ selects the value of field \abs{f} from heap h, $\mathsf{store}(h,\xabs{f},\mathsf{t})$ stores the value of $\mathsf{t}$ in field \abs{f} of heap $h$.
Finally, $\mathsf{anon}(h)$ is a heap that may be different from heap $h$. 

\begin{definition}[Semantics of Terms]
Let $I$ be a map from function names to functions and from predicate names to predicates.
Let $\beta$ be a map from logical variable to semantic values.
The evaluation of terms in a state \defstatepair is as follows
\begin{align*}
\eval{x}{\defstatepair,I,\beta} &::= \beta(x) \\
\eval{\xabs{v}}{\defstatepair,I,\beta} &::= \cased{ \state(v) & \text{ if }\xabs{v} \neq \xabs{heap}\\ \heap &\text{ if }\xabs{v} = \xabs{heap}}\\
\eval{f(\many{\mathsf{t}})}{\defstatepair,I,\beta} &::= I(f)(\many{\eval{\mathsf{t}}{\defstatepair,I,\beta}})
\end{align*}
\end{definition}
\begin{definition}[Semantics of Formulas]
Let $\defstatepair$ be a state without symbolic values. 
The satisfiability relation $\defstatepair,\beta \models \phi$ is defined as follows
\begin{align*}
\defstatepair,\beta,I \models p(\many{\mathsf{t}}) &\iff I(p)(\many{\eval{\mathsf{t}}{\defstatepair,I,\beta}})\\
\defstatepair,\beta,I \models \phi_1 \vee \phi_2&\iff \defstatepair,\beta,I \models \phi_1\text{ or }\defstatepair,\beta,I \models \phi_2\\
\defstatepair,\beta,I \models \neg\phi&\iff \defstatepair,\beta,I \models \phi \text{ does not hold}\\
\defstatepair,\beta,I \models  \exists~S~x.~\phi&\iff \text{ there is some \textbf{x} in $S$ such that}\defstatepair,\beta,I[x \mapsto \mathbf{x}] \models \phi\\
\defstatepair,\beta,I \models \mathsf{t} \doteq \mathsf{t}' &\iff \eval{\mathsf{t}}{\defstatepair,I,\beta} = \eval{\mathsf{t}'}{\defstatepair,I,\beta}
\end{align*}
\end{definition}
We use a fixed interpretation $I$, that maps each function symbol to its natural semantic counterpart and omit it from the evaluation. Concerning the heap functions we demand only the following connection axioms, for all heaps $h$, all fields \abs{f} and terms $\mathsf{t}$.
\[I(\mathsf{select})(I(\mathsf{store})(h,\xabs{f},\mathsf{t}),\xabs{f}) = \mathsf{t}\]
We furthermore assume that all logical variables are unique and that the type and number of parameters for functions and predicates is adhered too.
We use common abbreviations such as $\forall~S~x.~\phi$ for $\neg\exists~S~x.~\neg\phi$ and \textsf{True}. 
We shorten comparisons for terms of \xabs{Bool} types by writing, e.g., $i > j$ instead of $i > j \doteq \mathbf{True}$ and $\mathsf{select}(\xabs{heap},\xabs{f})$ with $\xabs{this.f}$.

\begin{example}
The following expresses that in a given state, there is a positive entry in the list stored in \xabs{this.f} and this entry is equal to the sum of the variable \xabs{v} and its index.
\[\exists i \in \mathbb{N}.~\xabs{this.f}[i] > 0 \wedge \xabs{v} + i \doteq \xabs{this.f}[i]\]
We remind that \abs{Int} is semantically mapped to the integers and thus $\xabs{v} + i$ is well-typed and well-defined.
\end{example}

\subsubsection{Local Monadic Second-Order Trace Logic (MSO)}
MSO expresses properties of local traces. Its models are local trace and the whole semantic domain, to allow to quantify over method names etc.
Additionally to standard MSO constructs, we use $[\traceterm] \doteq \mathsf{evt}$ to say that the event at position $\traceterm$ of the trace is equal to the event term $\mathsf{evt}$.
Similarly, $[\traceterm] \vdash \phi$ expresses that the state at position $\traceterm$ is a model for the FOS formula $\phi$.
Both these constructs evaluate to false, if the element at position $\traceterm$ is not an event, resp. state.
\begin{definition}[Syntax]
Let $p$ range over predicate symbols, $f$ over function symbols, $x$ over logical variable names and $S$ over sorts. 
As sorts we take all data types $\dType$ and additionally $\xabs{Fut},\xabs{Expr},\mathbb{N},\mathbf{I},\mathsf{M}$, where $\mathbf{I}$ is a subset of the integers, $\xabs{Fut}$ is the set of all futures of all types, \abs{Expr} the set of all well-typed expressions and $\mathsf{M}$ the set of all method names.
Formulas $\psi$. terms $\traceterm$ and event terms $\mathsf{evt}$ are defined by the following grammar.

\begin{align*}
\psi ::=& p(\many{\traceterm}) \sep \psi \vee \psi \sep \neg \psi \sep \traceterm \subseteq \traceterm \sep \exists x\in S.~\psi \sep \exists X \subseteq S.~\psi\sep [\traceterm] \doteq \mathsf{evt} \sep [\traceterm] \vdash \phi\\
\traceterm ::=& x \sep f(\many{\traceterm}) \\
\mathsf{evt} ::=&\
\invocev(\traceterm,\traceterm,\traceterm,\traceterm,\many{\traceterm}) \sep 
\invocrev(\traceterm,\traceterm,\traceterm,\many{\traceterm}) \sep
\resolvev(\traceterm,\traceterm,\traceterm,\traceterm) \\
\sep&\resolvrev(\traceterm,\traceterm, \traceterm, \traceterm) \sep
\awaitev(\traceterm,\traceterm,\traceterm, \traceterm) \sep
\awaitrev(\traceterm,\traceterm,\traceterm, \traceterm) \\
\sep&\suspev(\traceterm,\traceterm,\traceterm, \traceterm) \sep
\susprev(\traceterm,\traceterm,\traceterm, \traceterm) \sep \noev \sep \marker
\end{align*}

\end{definition}
We assume two predicates $\mathsf{isEvent}$ and $\mathsf{isState}$ with their obvious interpretation (for this predicate, we regard \marker as an event).
Furthermore, we assume a function symbol $\mathsf{singleton}(\traceterm)$ whose interpretation maps an element to a set containing only this element. We write $\traceterm \in \traceterm'$ for $\mathsf{singleton}(\traceterm) \subseteq \traceterm'$.
For each type of event, we assume a predicate that holds iff the given position is an event of that kind, for some parameters, e.g.,
\[\mathsf{is\resolvrev}(i) \iff \exists f \in \xabs{Fut}.~\exists o \in \mathsf{O}.~\exists m \in \mathsf{M}.~\exists v \in \xabs{Any}.~[i] \doteq \resolvrev(o,f,m,v)\]

\begin{definition}[Semantics of Terms]
The semantics of terms and event terms is straightforward. Note that MSO-terms do not contain program variables or fields, thus the semantics of terms depend only on the variable assignment and the function symbol interpretation.
\begin{align*}
\eval{x}{I,\beta} = \beta(x) \quad \eval{f(\many{\traceterm})}{I,\beta} &= I(f)(\many{\eval{\traceterm}{I,\beta}}) \qquad \eval{\noev}{I,\beta} = \noev \quad \eval{\marker}{I,\beta} = \marker\\
\eval{\invocev(\traceterm,\traceterm,\traceterm,\traceterm,\many{\traceterm})}{I,\beta} &=
\invocev(\eval{\traceterm}{I,\beta},\eval{\traceterm}{I,\beta},\eval{\traceterm}{I,\beta},\eval{\traceterm}{I,\beta},\many{\eval{\traceterm}{I,\beta}}) \\
\eval{\invocrev(\traceterm,\traceterm,\traceterm,\many{\traceterm})}{I,\beta} &=
\invocrev(\eval{\traceterm}{I,\beta},\eval{\traceterm}{I,\beta},\eval{\traceterm}{I,\beta},\many{\eval{\traceterm}{I,\beta}}) \\
\eval{\resolvev(\traceterm,\traceterm,\traceterm,\traceterm)}{I,\beta} &=
\resolvev(\eval{\traceterm}{I,\beta},\eval{\traceterm}{I,\beta},\eval{\traceterm}{I,\beta},\eval{\traceterm}{I,\beta})\\
\eval{\resolvrev(\traceterm,\traceterm, \traceterm, \traceterm)}{I,\beta} &=
\resolvrev(\eval{\traceterm}{I,\beta},\eval{\traceterm}{I,\beta}, \eval{\traceterm}{I,\beta}, \eval{\traceterm}{I,\beta})\\
\eval{\awaitev(\traceterm,\traceterm,\traceterm, \traceterm)}{I,\beta} &=
\awaitev(\eval{\traceterm}{I,\beta},\eval{\traceterm}{I,\beta},\eval{\traceterm}{I,\beta}, \eval{\traceterm}{I,\beta})\\
\eval{\awaitrev(\traceterm,\traceterm,\traceterm, \traceterm)}{I,\beta} &=
\awaitrev(\eval{\traceterm}{I,\beta},\eval{\traceterm}{I,\beta},\eval{\traceterm}{I,\beta}, \eval{\traceterm}{I,\beta}) \\
\eval{\suspev(\traceterm,\traceterm,\traceterm, \traceterm)}{I,\beta} &=
\suspev(\eval{\traceterm}{I,\beta},\eval{\traceterm}{I,\beta},\eval{\traceterm}{I,\beta}, \eval{\traceterm}{I,\beta})\\
\eval{\susprev(\traceterm,\traceterm,\traceterm, \traceterm)}{I,\beta} &=
\susprev(\eval{\traceterm}{I,\beta},\eval{\traceterm}{I,\beta},\eval{\traceterm}{I,\beta}, \eval{\traceterm}{I,\beta})
\end{align*}
\end{definition}

\begin{definition}[Semantics of Formulas]
The semantics of MSO-formulas is analogous to the one of FOS, except that the model is a local trace and the constructs specific to MSO and our extension:

\noindent\scalebox{0.9}{\begin{minipage}{\textwidth}
\begin{align*}
\trace,I,\beta &\models \exists X \subseteq S.~\psi \iff \text{ there is a subset $\mathbf{X}$ of $S$ such that }\trace,I,\beta[X \mapsto \mathbf{X}] \models \psi\\
\trace,I,\beta &\models \traceterm \subseteq \traceterm' \iff \eval{\traceterm}{I,\beta}  \subseteq \eval{\traceterm'}{I,\beta} \\
\trace,I,\beta &\models [\traceterm] \doteq \mathsf{evt}\iff 1 \leq \eval{\traceterm}{I,\beta} \leq |\trace| \wedge \trace[\eval{\traceterm}{I,\beta}] = \eval{\mathsf{evt}}{I,\beta}\\
\trace,I,\beta &\models [\traceterm] \vdash \phi\iff 1 \leq \eval{\traceterm}{I,\beta} \leq |\trace| \wedge \trace[\eval{\traceterm}{I,\beta}]\text{ is a state }\wedge \trace[\eval{\traceterm}{I,\beta}],I,\beta \models \phi
\end{align*}
\end{minipage}}

\end{definition}

Relativization~\cite{old} syntactically restricts a formula on a substructure of the original model. This substructure is defined by another formula.
\begin{definition}[Relativization]

\scalebox{0.85}{\begin{minipage}{\textwidth}
\begin{align*}
(\psi \wedge \psi')[y \in S \setminus \psi''(y)] &= \psi[y \in S \setminus \psi''(y)] \wedge \psi'[y \in S \setminus \psi''(y)]\\
(\neg \psi)[y \in S \setminus \psi'(y)] &= \neg(\psi[y \in S \setminus \psi'(y)])\\
\traceterm \subseteq \traceterm'[y \in S \setminus \psi(y)] &= \traceterm \subseteq \traceterm'\\
p(\many{\traceterm})[y \in S \setminus \psi(y)] &=p(\many{\traceterm})\\
([\traceterm] \doteq \mathsf{evt})[y \in S \setminus \psi(y)] &= [\traceterm] \doteq \mathsf{evt}\\
([\traceterm] \vdash \phi)[y \in S \setminus \psi(y)] &= [\traceterm] \vdash \phi\\
(\exists x \in S.~\psi)[y \in S' \setminus \psi'(y)] &= \cased{\exists x \in S.~\left(\psi'(x) \wedge \psi[y \in S' \setminus \psi'(y)]\right) & \text{ if }S=S'\\ \exists x \in S.~\psi[y \in S' \setminus \psi'(y)] & \text{ otherwise}}\\
(\exists X \subseteq S.~\psi)[y \in S' \setminus \psi'(y)] &= \cased{\exists X \subseteq S.~\left(\forall x \in X.~\psi'(x) \wedge \psi[y \in S' \setminus \psi'(y)]\right) & \text{ if }S=S'\\ \exists X \subseteq S.~\psi[y \in S' \setminus \psi'(y)] & \text{ otherwise}}
\end{align*}
\end{minipage}}

\end{definition}

\begin{example}
The following expresses that a trace contains only \noev events.
\[\psi = \forall i \in \mathbf{I}.~\mathsf{isEvent}(i) \rightarrow [i] \doteq \noev\]
The following expresses that the trace up to position $p$, a trace contains only \noev events.
\[\psi[j \in \mathbf{I} \setminus j < p] \equiv \forall i \in \mathbf{I}.~\left(i< p \rightarrow \left(\mathsf{isEvent}(i) \rightarrow [i] \doteq \noev\right)\right) \]
\end{example}

%% file: chapters/04_bpl.tex
\section{A Further Example}\label{sec:bpl1}
As described above, this technical report introduces Method Types for CAO without suspension to simplify presentation.
For illustration, we use the following examples.
\begin{example}\label{ex:original}
Fig.~\ref{fig:ex:met} shows a simple method that passes its input to \abs{Comp.cmp} and reads the result. If the result is negative,
its sign is inverted and the original input data is logged by \abs{Log.log}. The possibly inverted result is returned.
\end{example}
\begin{figure}
\begin{abscode}
class T(Comp S, Log L){
 Int test(Int i){
  Fut<Int> f = S!cmp(i);
  Int r = f.get$_0$;
  if(r < 0){
   r = -r;  f = L!log(i);
  }
  return r;
 }
}
\end{abscode}
\caption{An example method}
\label{fig:ex:met}
\end{figure}

\begin{example}\label{ex:mso}
Let $\xabs{r = f.get}_0$ be the synchronization from Ex.~\ref{ex:original}.
The following formula expresses that if every value read from a future of \abs{cmp} is positive, and every future read at synchronization point 0 is from \abs{cmp}, then in the state after the read, the value of \abs{r} is positive.
\begin{align*}
&(\forall i \in \textbf{I}.~\left(\forall v \in \xabs{Int}.~[i] \doteq \resolvrev(\_,\_,\xabs{cmp},v,\_) \rightarrow v > 0\right) \wedge\\
&\forall i \in \textbf{I}.~\left(\forall m \in \mathsf{M}.~[i] \doteq \resolvrev(\_,\_,m,\_,0) \rightarrow m \doteq \xabs{cmp}\right))\\
&\rightarrow \forall i\in\textbf{I}.~\left([i] \doteq \resolvrev(\_,\_,\_,\_,0) \rightarrow [i+1] \vdash \xabs{r > 0}\right)
\end{align*}
\end{example}

\section{Behavioral Program Logic}\label{sec:idea}
\input{types}
\section{A Sequent Calculus for BPL: Behavioral Types}\label{sec:symbol}
\input{newsimple.tex}

%% file: types.tex
Behavioral Program Logic (\TPL) is an extension of FOS with \emph{behavioral modalities} $[\statement \halfsim^\typesem \tau]$ that contain a statement $\statement$ and a behavioral specification $(\tau,\typesem)$. 
A behavioral specification consists of (1) a syntactic component (the type $\tau$) and (2) a translation $\alpha$ of the type into an MSO formula that has to hold for all traces generated by the statement. 
Behavioral specifications can be seen as representations of a certain class of MSO formulas, which are deemed useful for verification of distributed systems.
For the rest of this section, we assume fixed parameters \prgm, \objname, \future, \methodname for evaluation.
\begin{definition}[Behavioral Program Logic]
A \emph{behavioral specification} $\type$ is a pair $(\typesyn_\type, \typesem_\type)$, where $\typesem_\type$ maps elements of $\typesyn_\type$ to MSO formulas.

$\TPL$-formulas $\phi$, terms $t$ and updates $U$ are defined by the following grammar, which extends Def.~\ref{def:FOS}.
The meta variables range as in Def.~\ref{def:FOS}. Additionally let $\statement$ range over statements and $(\typesyn_\type, \typesem_\type)$ over behavioral specifications.
\\\noindent\scalebox{0.95}{
\begin{minipage}{\textwidth}
\begin{align*}
\phi &::=  \dots \sep [\statement \halfsim^{\typesem_\type} {\typesyn_\type}]\sep \{U\}\phi  \quad \mathsf{t} ::= \dots \sep \{U\}\mathsf{t} \quad U ::= \epsilon \sep U || U \sep \{U\}U \sep \xabs{v} := \mathsf{t}
\end{align*}
\end{minipage}}
\end{definition}
The semantics of a behavioral modality $[\statement \halfsim^{\typesem_\type} {\typesyn_\type}]$ is that all traces generated by \statement \emph{selected within \prgm} are models for $\typesem_\type(\typesyn_\type)$.
We use updates~\cite{keybook} to keep track of state changes, their semantics is a state transition. Update $\xabs{v} := \mathsf{t}$ changes the state by updating \abs{v} to $\mathsf{t}$.
The parallel update  $U || U'$ applies $U$ and $U'$ in parallel, with $U'$ winning in case of clashes. $\epsilon$ is the empty update and update application $\{U\}$ evaluates the term (resp. formula) in the state after applying $U$.

\begin{definition}[Semantics of \TPL]
The semantical extension of FOS to \TPL is given in Fig.~\ref{fig:bpl}. The interpretation $I$ has the properties described above. 
A formula $\phi$ is \emph{valid} if every $\defstatepair$ and every $\beta$ make it true.
\end{definition}
\begin{figure}[b!t]
\scalebox{0.9}{
\begin{minipage}{\textwidth}
\begin{align*}
&\eval{\{U\}t}{\defstatepair,I,\beta} = \eval{t}{\eval{U}{\defstatepair,I,\beta},I,\beta} \quad \eval{\epsilon}{\defstatepair,I,\beta}\left(x\right) = x \quad
\defstatepair,I,\beta \models \{U\}\phi \Leftrightarrow \eval{U}{\defstatepair,I,\beta},I,\beta \models \phi 
\\&\hspace{20mm}\eval{\xabs{v} := \mathsf{t}}{\defstatepair,I,\beta}\left(\defstatepairo\right) = 
\cased{\statepair{\state'}{\heap''} &\text{ if \abs{v} = \abs{heap}, }\heap'' = \eval{\mathsf{t}}{\defstatepair,I,\beta}\\
\statepair{\state''}{\heap'} & \text{ otherwise, }\state'' = \state'[\xabs{v} \mapsto \eval{\mathsf{t}}{\defstatepair,I,\beta}]
}
\\
&\eval{U || U'}{\defstatepair,I,\beta}\left(x\right) = \eval{U'}{\defstatepair,I,\beta}\left(\eval{U}{\defstatepair,I,\beta}\left(x\right)\right) \qquad 
\eval{\{U\} U'}{\defstatepair,I,\beta} = \eval{U'}{\eval{U}{\defstatepair,I,\beta}I,\beta} \\
&\hspace{30mm}\defstatepair,I,\beta \models [\statement \halfsim^{\typesem_\type} {\typesyn_\type}] \Leftrightarrow \forall \trace \in \eval{\statement}{\objname,f,\methodname,\defstatepair}^\prgm.~\trace,I,\beta\models\typesem_{\type}(\typesyn_{\type})
\end{align*}
\end{minipage}
}
\caption{Semantics of \TPL. The satisfiability relation on the right of the semantics of behavioral modalities is the one of MSO.}
\label{fig:bpl}
\end{figure}

Object, program, method name, resolved future and type of \abs{result} are implicitly known, but we omit them for readability's sake.
We use a sequent calculus to reason about \TPL (resp. FOS).
\begin{definition}[Sequents and Rules]
Let $\Delta, \Gamma$ be sets of \TPL-formulas. A \emph{sequent} $\Gamma \Rightarrow \Delta$ has the semantics of $\bigwedge \Gamma \rightarrow \bigvee \Delta$.
$\Gamma$ is called the antecedent and $\Delta$ the succedent.
Let $C,P_i$ be sequents. A rule has the form 
\TINFER[cond]{name}{P_1 \quad\dots\quad P_n}{C}
Where $C$ is called the conclusion and $P_i$ the premise, while $cond$ is a side-condition. Side-conditions are always decidable. For readability's sake, we apply side conditions containing equalities directly in the premises. 
\end{definition}
\noindent Rules may contain, in addition to expressions, schematic variables. Their handling is standard~\cite{keybook}.
We assume the usual FO rules for the FOS part of \TPL handling all FO operators such as quantifiers. 
\begin{definition}[Soundness]
A rule is \emph{sound} 
if validity of all premisses implies validity of the conclusion.
\end{definition}
Soundness implicitly refers to a program \prgm, as behavioral modalities are defined over $\prgm$-selectable traces.
Rewrite rules $\typesyn_1 \leftrightsquigarrow \typesyn_2$ syntactically replace one type $\typesyn_1$ by another, $\typesyn_2$ (and vice versa) and are sound if $\typesem(\typesyn_1) \equiv \typesem(\typesyn_2)$.

\paragraph{Discussion.} Before we introduce method types, a particular behavorial specification, we illustrate \TPL with further examples.
To reason about postconditions, as standard modal logics, we define a behavioral specification that only uses the last state of a trace for its semantics.
\begin{example}
The specification for postconditions is the pair of the set of all FOS sentences and the function \posttypesem, defined below. 
\textsf{T} is the type of \abs{result}.

\noindent\scalebox{0.85}{\begin{minipage}{0.9\textwidth}\[
\posttypesem(\phi) =\cased{
 \exists v \in \mathsf{T}.~[\mathit{last}\!-\!1] \doteq \resolvev(\_,\_,\_,v) \wedge [\mathit{last}] \vdash \phi[\xabs{result}\!\setminus\!v]
 &\text{ if $\phi$ contains \abs{result}}
\\{[\mathit{last}]} \vdash \phi &\text{ otherwise}}
\]\end{minipage}}
\end{example}
A Hoare triple $\{\phi\}\statement\{\psi\}$ has the same semantics as the formula $\phi \rightarrow [\statement \halfsim^\posttypesem \psi]$.
A standard dynamic logic modality $[\statement]\psi$ has the same semantics as the behavioral modality $[\statement \halfsim^\posttypesem \psi]$
\footnote{This justifies our use of the term ``modality''. Contrary to standard modalities, behavioral modalities are not formulas that express modal statements about \emph{formulas}, but formulas that express a modal statement about \emph{more general specifications}.}. 
Behavioral modalities generalize these systems and can be used to express any trace property (expressible in MSO), independent of the form of its verification system. 
The following defines a points-to analysis (for the next statement), normally implemented in a data-flow framework.
\begin{example}[Points-To]\label{ex:ptpost}
The behavioral specification of a \emph{points-to analysis} specifies that the next statement reads a future resolved by a method from set $M$. 
\\\noindent\scalebox{0.85}{\begin{minipage}{\textwidth}\begin{align*}
\pttype &= (\mathcal{P}(\method),\pttypesem) \text{ with }\\
\pttypesem(M) &= \exists \objname\in \mathsf{O}.\exists f\in\xabs{Fut}.\exists \methodname\in\mathsf{M}.\exists v \in \mathsf{Any}.\exists i\in\mathbb{N}.~
[1] \doteq \resolvrev(\objname,f,\methodname,v,i) \wedge \bigvee_{\methodname' \in M} \methodname \doteq \methodname'
\end{align*}\end{minipage}}

The following formula expresses that the \abs{get} statement reads a positive number, if the future is resolved by \abs{Comp.cmp}.
This is the case if $\xabs{Comp.cmp}$ always returns positive values. 
The identifier connects the two modalities semantically.
\[\phi_p = [\xabs{r = f.get}_0 \halfsim^\pttypesem \{\xabs{Comp.cmp}\}] \rightarrow [\xabs{r = f.get}_0 \halfsim^{\posttypesem} \xabs{r}>0]\]
%

It is not necessary to include postcondition reasoning. 
Rule \rulename{ex1} in Fig.~\ref{fig:bsrules} expresses that if it is known that the next read from \abs{s} is from some set $E'$ and it is required to show that the next read is from $E$, it suffices to check whether $E$ is a subset of $E'$.
Similarly, rule \rulename{ex-$\halfsim$} connects two analyses: one may assume some formula $\psi$ for a read value, if one can show that this synchronization always reads from method \abs{Comp.cmp} and that $\statement_\xabs{Comp.cmp}$, the method body of \abs{Comp.cmp} establishes $\psi$. This corresponds to a generalization of Ex.~\ref{ex:mso}. 
\end{example}
\begin{figure}[tbh]
\scalebox{0.8}{\begin{minipage}{0.5\textwidth}
\SAINFER{$E \subseteq E'$}{ex1}{
$\Gamma, [\xabs{s} \ptsim E']\Rightarrow [\xabs{s} \ptsim E], \Delta$
}    
\end{minipage}}
\scalebox{0.8}{\begin{minipage}{0.5\textwidth}
\SSINFER[$v$ fresh]{ex-$\halfsim$}{
    \Gamma, \psi(v) \fCenter \{\xabs{v} := v\}[\xabs{s} \postsim  \phi],\Delta}{
    \fCenter[\xabs{v = f.get}_0 \ptsim \{\xabs{Comp.cmp}\}]\wedge [\statement_\xabs{Comp.cmp} \postsim \psi]
}{
\Gamma    \fCenter [\xabs{v = f.get}_0\xabs{;s} \postsim  \phi], \Delta
}
\end{minipage}}
\caption{Two example rules for behavioral specifications. $\psi(v)$ replaces \abs{result} by $v$ and we assume that $\psi$ contains no fields.}
\label{fig:bsrules}
\end{figure}

The above example illustrates the difference between modalites and typing judgments. Modalites are formulas and can be used for deductive reasoning about a type judgment (which, in our case, is encoded into $\halfsim$).
While a calculus for $\posttypesem$ is easily carried over from other sequent calculi, this is not possible for all behavioral specifications. 
The proof can still be closed in two ways.
\begin{itemize}
\item There may be some rules, such as \rulename{ex1} above, that enable to reason about the analysis without reducing the statement at all.
\item If the proof contains only open branches containing behavioral specification, one may run a static analysis to evaluate them to true or false directly. E.g., if for the formula $\phi_p$ above the pointer analysis returns that the synchronization
point 0 reads from \abs{L.log}, the first behavioral modality evaluates to false and the whole formula to true.
\end{itemize}
Using external analyses increases modularity: (1) the \TPL-calculus is simpler because it does not need to encode the implementation and (2) one may verify functional correctness of a method \emph{up to its context}.
Open branches are then a description of the context which the method requires. 
This may be verified once more context is known, thus extending proof repositories~\cite{repos} to external analyses.


\COMMENT{
Fig.~\ref{fig:logics} shows how FOS, MSO, behavioral modalities and \TPL build upon each other: FOS describes single states in traces of the locally abstract semantics of statements, MSOT describes single traces,
behavioral modalities describe multiple selectable traces by giving them a logical characterization as MSO formulas and \TPL uses multiple behavioral modalities to reason about multiple characterizations of statements.
\begin{figure}
\centering\includegraphics[scale=0.5]{drawing.pdf}
\caption{Locally abstract semantics, FOS, MSO, behavioral modalities and \TPL.} 
\label{fig:logics}
\end{figure}
}

%% file: newsimple.tex
IIn this section we characterize behavioral types as behavioral specifications with a set of sequent calculus rules and a constraint on the proof obligations of the methods within a program.
Before we formalize this in general, we introduce method types~\cite{ifm,icfem}, a behavioral type for Active Objects that suffices to generalize method contracts and object invariants by integrating the behavioral specifications for postcondition reasoning and points-to analysis. The method type of a method describes the local view of a method on a protocol.
\begin{definition}
The local protocol $\mathbf{L}$ of a method and the method type $\methtype$ are defined by the grammar below.
The behavioral specification for method types is $\localtype = (\methtype,\localtypesem)$. 
Let $\objname_0,\dots,\objname_n$ be roles, and $\xabs{f}_{\objname_0},\dots,\xabs{f}_{\objname_n}$ fields of fitting type.
\localtypesem is defined as $\exists \objname_0,\dots,\objname_n\in\mathsf{O}.~\bigwedge_{i\leq n}\objname_i \doteq \xabs{f}_{\objname_i} \wedge \localtypesem'(\localtype)$.
\\\noindent\scalebox{0.9}{
\begin{minipage}{\textwidth}
\begin{align*}
\mathbf{L} ::= ?\methodname(\phi).\methtype \qquad\qquad
\methtype &::= 
\objname!\methodname(\phi) \sep \downarrow\!(\phi) \sep \mathsf{skip} \sep \methtype.\methtype \sep \methtype^\ast \sep \oplus\{{\methtype}_i\}_{i\in I} \sep \&(\many{\methodname},\phi)\{\methtype,\methtype\} 
\end{align*}
\end{minipage}
}
\end{definition}
The local protocol of a method contains the receiving action $?\methodname(\phi)$, which models that the parameters satisfy the predicate $\phi$.
The method body is checked against the method type -- there is no statement corresponding to receiving.
Roles keep track of an object through the protocol. The assignment of roles to fields is generated by the projection of session types~\cite{icfem}.
We stress that statements and method types share syntactic elements -- it is possible to pattern match on statements/expressions on one side and a method type on the other side in rules. 

Calls are specified with the call action $\objname!\methodname(\phi)$, where $\objname.\methodname$ is the receiver and the predicate $\phi$ has to hold. Here, $\phi$ does not only specify the sent data but also 
local variables and fields. It can express properties such as ``the sent data is larger then some field''. The termination action $\downarrow\!(\phi)$ models termination in a state satisfying $\phi$ (which again may include \abs{result}).
The empty action $\mathsf{skip}$ models no visible actions and $\methtype_1.\methtype_2$ to sequential composition: all interactions in $\methtype_1$ must happen before $\methtype_2$.
Repetition $\methtype^\ast$ corresponds to the Kleene star (and loops) and models zero or more repetitions of the interactions in $\methtype$.

There are two choice operators: $\oplus\{{\methtype}_i\}_{i\in I}$ is the active choice, the method must select one branch $\methtype_i$. 
It is not necessary to implement all branches, the method may choose to never select some branches. The index set $I$ must not be empty.
$\&(\many{\methodname},\phi)\{\methtype_1,\methtype_2\}$ is the passive choice: some other method made a choice and this method has to follow the protocol according to this choice.
The choice is communicated via a future which has to be resolved by one of the methods in $\many{\methodname}$.
If the choice condition $\phi$, which may only include the program variable \abs{result}, is fulfilled by the read data, $\methtype_1$ has to be followed, otherwise $\methtype_2$ has to be followed.
Both branches have to be implemented. 

The semantics of the call and termination actions specify a trace with at least three elements with the correct event on second position and a state fulfilling the given predicate on the third position. Every other event is \noev.
The semantics of the empty action and active choice are straightforward. Sequential composition uses relativization: some position $i$ is chosen, such that the left translation holds before $i$
and the right translation afterwards. Note that $i$ is included in both relativization, to uphold the invariant that a trace always starts and ends with a state.
The semantics of repetition are the only point where we require second order quantifiers:
set $I$ is a set of indices, such that the first and last position are included and for every consecutive pair $k,l$ of elements of $I$,
the translation of the repeated type holds in the relativization between $k$ and $l$. 
Passive choice specifies that the first event is a read on a correct future (i.e., resolved by the correct method) and the suffix afterwards follows the communicated choice correctly.


\begin{figure}
\caption{Semantics for \localtype. Unbound variables are implicitly existentially quantified.}
\scalebox{0.9}{
\begin{minipage}{\textwidth}
\begin{align*}
\localtypesem'(\objname!\methodname(\phi)) &= \forall i \in \mathbf{I}.~\mathsf{isEvent}(i) \wedge [i] \not\doteq \noev \rightarrow [i] \doteq \invocev(x,\objname,f,\methodname,\many{\expr}) \wedge [i-1] \vdash \phi(\many{\expr})\\
\localtypesem'(\downarrow\!(\phi)) &= \forall i \in \mathbf{I}.~\mathsf{isEvent}(i) \wedge [i] \not\doteq \noev \rightarrow [i] \doteq \resolvev(x,f,m,\expr) \wedge [i-1] \vdash \phi[\xabs{result} \setminus \expr]\\
&\text{where $\phi(\many{\expr})$ replaces its free variables by $\many{\expr}$. $\phi[\xabs{result} \setminus \expr]$ replaces $\xabs{result}$ by $\expr$.}\\
\localtypesem'(\mathsf{skip}) &= \forall l \in \mathbf{I}.~[l]\doteq\noev \vee [l] \vdash \xabs{true}\qquad
\localtypesem'(\oplus\{{\methtype}_i\}_{i\in I}) = \bigvee_{i\in I} \localtypesem'(\methtype_i)\\
\localtypesem'(\methtype_1.\methtype_2) &= \exists i\in\mathbf{I}.~\localtypesem'(\methtype_1)[n \in \mathbf{I} \setminus n\leq i] \wedge \localtypesem'(\methtype_2)[n \in \mathbf{I} \setminus n\geq i] \\
\localtypesem'(\methtype^\ast) &= \exists I \subseteq \mathbf{I}.~\exists a,b \in I.~ a < b \wedge \\
&\phantom{\exists I \subseteq \mathbf{I}.~}\forall k\in \mathbf{I}.~\left((k < a \wedge \mathsf{isEvent}(i) \rightarrow [i] \not\doteq \noev)\vee (a \leq k \wedge k \leq b)\right)\wedge\\
&\phantom{\exists I \subseteq \mathbf{I}.~}\forall i_1,i_2 \in I.~\left((\forall l \in I.~ l \leq i_1 \wedge i_2 \leq l) \rightarrow \localtypesem'(\methtype)[n \in \mathbf{I} \setminus i_1 \leq n \wedge n \leq i_2)]\right)\\
\end{align*}
\vspace{-12mm}
\begin{align*}
\localtypesem'(&\&(\{\methodname_l\}_{l\in I},\phi)\{\methtype_1,\methtype_2\}) = \exists i,j,k\in\mathbf{I}.~i<j\wedge j<k \wedge \\
&(\forall l \in \mathbf{I}.~l\doteq j \vee l\geq k\vee(l \leq i \wedge ([l] \doteq \noev \vee [l] \vdash \xabs{true})) \wedge [j] \doteq \resolvrev(x,m,f,\expr,n) \wedge \\
&\bigvee_{l\in I}m \doteq \methodname_i  \wedge ([k] \vdash \phi \rightarrow \localtypesem'(\methtype_1)[n \in \mathbf{I} \setminus n\geq k]) \wedge ([k]\not\vdash \phi \rightarrow \localtypesem'(\methtype_2)[n \in \mathbf{I} \setminus n\geq k])
\end{align*}
\end{minipage}}
\label{fig:semlocaltype}
\end{figure}


\begin{example}\label{ex:mytype}
The following formalizes the behavior described informally in Ex.~\ref{ex:mso}:
\scalebox{0.76}{\begin{minipage}{\textwidth}
\begin{align*}
?\xabs{T.test}(\mathsf{true}).\xabs{S}!\xabs{Comp.cmp}(\xabs{data}\doteq\xabs{i}).\localpaschoice{\{\xabs{Comp.cmp}\},\xabs{result}<0}{\xabs{L}!\xabs{Log.log}(\xabs{data}\doteq\xabs{i}),\\\mathsf{skip}}.\downarrow\!(\xabs{result} \geq 0)
\end{align*}
\end{minipage}}

The \abs{result} variable in the guard of the passive choice is referring to the result of the read value, not the specified method.
\end{example}

We define behavioral types from a program logic perspective\footnote{Behavioral types are sometimes (informally) distinguished from data types
by having a subject reduction theorem where the typing relation is preserved,
but not the type itself~\cite{dezanimail}. In \TPL this would correspond to the property
that one of the rules has a premise where the type in the
behavioral modality is different than in the conclusion.} by a type system, which is a set of sequent calculus rules that match on behavioral modalities and a obligation scheme, that maps every method to a proof obligation
\begin{definition}[Behavioral Types]~\label{def:behtype}
A \emph{behavioral type} $\type$ is a behavioral specification $(\typesyn_\type, \typesem_\type)$ extended with $(\typecalc_\type, \typescheme_\type)$.

 The \emph{obligation scheme} $\typescheme_\type$ maps method names \methodname to \emph{proof obligations}, sequents of the form $\phi_\methodname \Rightarrow [ \statement_\methodname \halfsim^{\typesem_\type} \typesyn_\methodname]$, which have to be proven. $\statement_\methodname$ is the method body of \methodname. 
The \emph{type system} $\typecalc_\type$ is a set of rewrite rules for $\typesyn_{\type}$ and sequent calculus rules with conclusions matching the sequent $\Gamma \fCenter \{U\}[\statement \halfsim^{\typesem_\type} \typesyn_\type], \Delta$.
\end{definition}
We demand that obligation schemes are consistent, i.e., proof obligations do not contradict each other.
This would be the case if, for example a method is called and its precondition $\phi$ is checked caller-side, then $\phi$ must truly be used as a precondition by
the proof obligation for the called method.
\begin{definition}
Let $\mathbf{L}_{\methodname} = ?\methodname(\phi_{\methodname}).\methtype_{\methodname}$ be the local protocols in $\prgm$.
We require that all $\mathbf{L}_{\methodname}$ are consistent: If $\methodname$ is called in the method type of any other method $\methodname'$,
then the call condition implies $\phi_{\methodname}$. Furthermore, $\phi_{\xabs{X.run}} = \xabs{true}$.

The extension of the behavioral specification \localtype of method types to a behavioral type is given by the calculus in Fig.~\ref{fig:totalrules} and 
$\localtypescheme(\methodname) = \phi_\methodname \Rightarrow [ \statement_\methodname \halfsim^{\localtypesem} \methtype_\methodname]$. 
\input{newtotalrules}

\end{definition}
The call condition may contain fields of the other objects, but this is not an issue when checking consistency, as the precondition only contains fields of the own object and the fields are simply uninterpreted function symbols.
The method in Fig.~\ref{fig:ex:met} can be typed with the type in Ex.~\ref{ex:mytype}. 

Rule \rulename{\textsf{met}--V} translates a variable-assignment into an update and \rulename{\textsf{met}--F} is analogous for fields.
Rule \rulename{\textsf{met}--\abs{get}} has three premises: one premise checks via \pttype that the correct methods are synchronized with.
The two others use a fresh constant $v$ for the read value and assign it to the target variable. The two premises differ in the branch that is checked afterwards, depending on whether or not
the choice condition holds.
Rule \rulename{\textsf{met}--\abs{while}} is a standard loop invariant rule. An invariant $I$ holds before the first iteration and is preserved by the loop to remove all other information afterwards. 
The loop body is checked against the repeated type and the continuation against the continuation of the type. Method types have no special action for the end of a statement, so \posttype is used for checking that the loop preserves its invariant.
Rule \rulename{\textsf{met}--\abs{if}} splits the set of possible choices into two and checks each branch against one of these sets. These sets may overlap and do not need to cover all original choices, but may not be empty.
Rule \rulename{\textsf{met}--call} checks the annotated condition of the called method and the correct target explicitly and that the correct method is called by matching call type and call statement. 
We remind that references are not reassigned, so call targets can be verified locally.
The other rules are straightforward.

\paragraph{Contracts and Invariants.}
Method types generalize method contracts and object invariants as follows. An object invariant is encoded by adding it to the formula in the receiving and terminating actions of all method in an object -- except the constructor \abs{run}, where it is only added to the terminating action.
A method contract (consisting of a precondition on the parameters and a postcondition) is encoded analogously by adding the precondition to the receiving and the postcondition to the terminating actions.
However, one additional step is required: Method types are generated by projection of global types~\cite{ifm}, so to use them for object invariants or method contracts requires to infer a method type first.
This is done by mapping every call to a call action, every branching to an active choice, every loop to a repetition, termination to a terminating action
and using $\mathsf{true}$ at every position where a formula is required, before adding precondition, postcondition or object invariant.
The most complex construct is synchronization. Each such read is mapped to a passive choice with all methods as the method set and $\mathsf{true}$ as the choice condition.
The following code is added in the first branch. The second branch is $\mathsf{skip}$. 
Invariants require fields in the precondition and a fitting notion of consistency, which was developed in~\cite{ifm}. 

\begin{example}
Consider the code in Fig.~\ref{fig:ex:class}, a variation of our running example.
It tracks the number of calls to \abs{T.test} and inverts the result if the input is positive.
It adheres to the contract with precondition $\xabs{i} \geq 0$ and postcondition $\xabs{result} \geq 0$ and the invariant
$\xabs{this.nr} \geq 0$.
Our algorithm derives the following type:

\hspace{5mm}\scalebox{0.8}{
\begin{minipage}{\textwidth}
\begin{align*}
?\xabs{T.test}(\mathsf{true})~.~\xabs{S}!\xabs{Comp.cmp}(\mathsf{true})~.~\localpaschoice{\mathsf{M},\mathsf{true}}{\localactchoice{\xabs{L}!\xabs{Log.log}(\mathsf{true}),\\\mathsf{skip}},\\\mathsf{skip}}.\downarrow\!(\mathsf{true})
\end{align*}
\end{minipage}}

\noindent 
Let $\phi_{\xabs{cmp}}$ and $\phi_{\xabs{log}}$ be the preconditions of the called methods.
The final specification, after adding the contract and the invariant, is shown on the right in Fig.~\ref{fig:ex:class} as $\methtype_1$.
The inferred type is not the one we gave in Ex.~\ref{ex:mytype}:
For one, it differs in its shape (two choice operators). 
For another, it neither keeps track of the passed data, nor specifies the relation between the return value of \abs{Comp.cmp} and the taken branch. 
These properties are typical for protocol specifications and require a global view, contrary to the local view of method contracts and object invariants.
However, one can add the pre- and postcondition and the object invariant also to the type given in Ex.~\ref{ex:mytype} and combine local and global specification.
The result is shown as $\methtype_2$ in in Fig.~\ref{fig:ex:class}. $\methtype_2$ expresses that the method follows the protocol and adheres to contract and object invariant. 
\begin{figure}[bt]
\begin{minipage}{0.45\textwidth}
\begin{abscode}
class T(Comp S, Log L){
 Int nr = 0;
 Int test(Int i){
  Fut<Int> f = S!cmp(i);
  this.nr = this.nr + 1;
  Int r = f.get$_0$;
  if(r < 0 && i > 0){
   r = -r; f = L!log(i);
  }
  return r;
 }
}
\end{abscode}
\end{minipage}
\scalebox{0.8}{
\begin{minipage}{0.55\textwidth}
\begin{tabular}{rl}
Precondition:& $\xabs{i} \geq 0$\\
Postcondition:& $\xabs{result} \geq 0$\\
Invariant:& $\xabs{this.nr} \geq 0$\\
\end{tabular}
\scalebox{0.85}{
\begin{minipage}{\textwidth}
\begin{align*}
\methtype_1=~~&?\xabs{T.test}(\xabs{i} \geq 0 \wedge \xabs{this.nr} \geq 0).\xabs{S}!\xabs{Comp.cmp}(\phi_{\xabs{cmp}})\\
.~&\localpaschoice{\mathsf{M},\mathsf{true}}{\localactchoice{\xabs{L}!\xabs{Log.log}(\phi_{\xabs{log}}),\\\mathsf{skip}},\\\mathsf{skip}}\\
.~&\downarrow\!(\xabs{result} \geq 0 \wedge \xabs{this.nr} \geq 0)\\[4mm]
\methtype_2=~~&?\xabs{T.test}(\xabs{i} \geq 0 \wedge \xabs{this.nr} \geq 0).\xabs{S}!\xabs{Comp.cmp}(\xabs{data}\doteq\xabs{i}\wedge\phi_{\xabs{cmp}})\\
.~&\localpaschoice{\{\xabs{Comp.cmp}\},\xabs{result}<0}{\xabs{L}!\xabs{Log.log}(\xabs{data}\doteq\xabs{i}\wedge\phi_{\xabs{log}}),\\\mathsf{skip}}\\
.~&\downarrow\!(\xabs{result} \geq 0) \wedge \xabs{this.nr} \geq 0
\end{align*}
\end{minipage}}
\end{minipage}}
\caption{An example method and two method types for method contracts and invariants.}
\label{fig:ex:class}
\end{figure}
\end{example}

\begin{theorem}\label{lem:local}
\localtype is sound for every program and consistent obligation scheme.
\end{theorem}
\begin{proof}
We need to show the soundness of all rules.
For convenience, we omit the quantifiers and assignment of roles to fields, which is trivial for all but one rule.
\begin{itemize}
\item \textbf{\textsf{met}-V} By definition of soundness we may assume that for all $\defstatepair,\beta$
\[
\defstatepair,I,\beta \models \bigwedge\Gamma \rightarrow \{U\}\{\xabs{v := e}\}[\statement \localsim \methtype] \vee \bigvee \Delta
\]
By semantics of the implication, we may assume for all $\defstatepair,\beta$ that if $\defstatepair,I,\beta \models \bigwedge\Gamma$ holds, then
\begin{equation}
\defstatepair,I,\beta \models\{U\}\{\xabs{v := e}\}[\statement \localsim \methtype] \text{ or } \defstatepair,I,\beta \models \bigvee \Delta \label{pr:1}
\end{equation}
We use the following abbrevations 
\begin{align*}
\statepair{\state'}{\heap'} &= \eval{U}{\defstatepair,I,\beta}\left(\defstatepair\right)\\
\statepair{\state''}{\heap''} &= \eval{\xabs{v := e}}{\statepair{\state'}{\heap'},I,\beta}\left(\defstatepair\right) = \statepair{\state'[\xabs{v} \mapsto \xabs{e}]}{\heap'}
\end{align*}
and unroll the semantics of the premise as follows. 
\begin{align}
&\defstatepair,I,\beta \models\{U\}\{\xabs{v := e}\}[\statement \localsim \methtype]\nonumber\\
\iff &\statepair{\state'}{\heap'},I,\beta \models \{\xabs{v := e}\}[\statement \localsim \methtype]\nonumber\\
\iff &\statepair{\state''}{\heap''},I,\beta \models [\statement \localsim \methtype]\nonumber\\
\iff &\forall \trace \in \eval{\statement}{\objname,\future,\methodname,\statepair{\state''}{\heap''}}^\prgm.~ \trace,I,\beta \models \localtypesem(\methtype)\label{pr:2}
\end{align}

We need to show that for all $\defstatepair,\beta$
\[
\defstatepair,I,\beta \models \bigwedge\Gamma \rightarrow \{U\}[\xabs{v = e;}\statement \localsim \methtype] \vee \bigvee \Delta
\]
Again, by semantics of the implication, we need to show for all $\defstatepair,\beta$ that if $\defstatepair,I,\beta \models \bigwedge\Gamma$ holds, then
\[
\defstatepair,I,\beta \models\{U\}[\xabs{v = e;}\statement \localsim \methtype] \text{ or } \defstatepair,I,\beta \models \bigvee \Delta
\]
The second case is obvious from \ref{pr:1}. It remain to show that for all $\defstatepair,\beta$ the statement~\ref{pr:2} implies the following
\begin{align}
&\defstatepair,I,\beta \models\{U\}[\xabs{v = e;}\statement \localsim \methtype]\nonumber\\
\iff &\statepair{\state'}{\heap'},I,\beta \models [\xabs{v = e;}\statement \localsim \methtype]\nonumber\\
\iff &\forall \trace' \in \eval{\xabs{v = e;}\statement}{\objname,\future,\methodname,\statepair{\state'}{\heap'}}^\prgm.~ \trace',I,\beta \models \localtypesem(\methtype)\label{pr:3}
\end{align}
Let $\trace = \selection \triangleright \history$. We observe that the first element of \history is $\statepair{\state'[\xabs{v} \mapsto \xabs{e}]}{\heap'} = \statepair{\state''}{\heap''}$ and
\[\trace' = \selection \triangleright \sequence{\statepair{\state'}{\heap'},\noev} \circ \history \]

We show that \ref{pr:2} implies \ref{pr:3} by induction on \methtype, i.e. for all traces $\trace',\trace$ as above the following holds:
\begin{equation}
\forall \methtype.~\trace,I,\beta \models \localtypesem(\methtype) \rightarrow \trace',I,\beta \models \localtypesem(\methtype)\label{intlem}
\end{equation}

In the following the notation $\trace[1..n]$ refers to the history of $\trace$.
\begin{itemize}
\item \emph{Base Case $\methtype = \mathsf{skip}$:} If $\trace$ only contains \noev events, then so does $\trace'$.
\item \emph{Base Case $\methtype = \objname!\methodname(\phi)$:} We observe that the translation of this type demands that there is exactly one non-\noev event. If this is the case for $\trace$, then it is also the case for $\trace'$.
\item \emph{Base Case $\methtype = \downarrow\!(\phi)$:} Analogous to the above case.
\item \emph{Step Case $\methtype = \methtype^\ast$:} We observe that the translation of this type allows a prefix consisting only of \noev-events.
\item \emph{Step Case $\methtype = \oplus\{{\methtype}_i\}_{i\in I}$:} If \trace is a model for $\oplus\{{\methtype}_i\}_{i\in I}$ then there is a $k\in I$ such that \trace is a model for $\methtype_k$. By induction hypothesis, $\trace'$ is also a model for $\methtype_k$.
\item \emph{Step Case $\methtype = \&(\many{\methodname},\phi)\{\methtype',\methtype''\}$:} We observe that the translation of this type allows a prefix consisting only of \noev-events.
\item \emph{Step Case $\methtype = \methtype'.\methtype''$:} The translation choses some $i$ such that $\trace[1..i]$ is a model for $\methtype'$ and $\trace[i..|\trace|]$ is a model for $\methtype''$. 
For $\trace'$ we chose the same position $i$. The subtrace $\trace'[i..|\trace'|]$ is a model for $\methtype''$, because 
\[
\trace[i..|\trace|] = \trace'[i..|\trace'|]
\]
It remains to show that $\trace'[1..i]$ is a model for $\methtype'$. We use the above observation
\[
\trace'[1..i] = \sequence{\statepair{\state'}{\heap'},\noev} \circ \trace[1..i] 
\]
Now by induction hypothesis, if $\trace[1..i]$ is a model for $\methtype'$, then so is $\trace'[1..i]$.
\end{itemize}
\item\textbf{\textsf{met}-F} Analogous to the above case.
\item\textbf{\textsf{met}-\abs{skip}} 
In this case, it suffices to show that 
\begin{align*}
&\forall \trace \in \defeval{\xabs{skip}}.~\trace,\beta,I \models \localtypesem(\mathsf{skip})\\
\iff&\emptyset \triangleright \sequence{\defstatepair},\beta,I \models  \forall l \in \mathbf{I}.~[l]\doteq\noev \vee [l] \vdash \xabs{true}
\end{align*}
This is obviously true.
\item\textbf{\textsf{met}-\abs{return}} 
In this case, it suffices to show that if $\{\xabs{result := e}\}\phi$ is valid, then
\\\noindent\scalebox{0.9}{\begin{minipage}{\textwidth}
\begin{align*}
&\forall \trace \in \defeval{\xabs{return e}}.~\trace,\beta,I \models \localtypesem(\downarrow(\phi))\\
\iff&\emptyset \triangleright \sequence{\defstatepair,\resolvev(X,f,\methodname,\expr),\defstatepair},\beta,I \models\\ 
&\forall i \in \mathbf{I}.~\mathsf{isEvent}(i) \wedge [i] \not\doteq \noev \rightarrow [i] \doteq \resolvev(X,f,\methodname,\expr) \wedge [i+1] \vdash \phi[\xabs{result} \setminus \expr]\\
\iff&\defstatepair,\beta,I \models \phi[\xabs{result} \setminus \expr] \iff \defstatepair,\beta,I \models \{\xabs{result := e}\}\phi
\end{align*}
\end{minipage}}

Which is valid by assumption. The last state is due to the semantics of updates.
\item\textbf{\textsf{met}-call} 
We may assume that $\phi(\expr)$ and $\{\xabs{v :=} f'\}[\statement  \localsim \methtype]$ hold in every state \defstatepair and that the object adheres to the assigment of roles to fields.
Analogous to the first case, it is required to show that in this case the following holds.
\\\noindent\scalebox{0.9}{\begin{minipage}{\textwidth}
\begin{align*}
&\forall \trace \in \eval{\xabs{v = f!m'(e);s}}{\objname,f,\methodname,\defstatepair}.~ \trace,\beta,I \models \localtypesem'(\xabs{X!m'(}\phi\xabs{)}.\methtype)\iff\\
&\forall \trace' \in\eval{\xabs{s}}{\objname,f,\methodname,\statepair{\state[\xabs{v} \mapsto f']}{\heap}}.~ \emptyset \triangleright \sequence{\defstatepair,\invocev(\objname,\xabs{O'},f',\xabs{m'},\expr)}\circ \theta',\beta,I \models \localtypesem(\xabs{O!m'(}\phi\xabs{)}.\methtype)
\end{align*}
\end{minipage}}

The assigment of roles to fields is adhered to by the second clause of the first premisse.
For the translation of the sequential composition we chose $i = 3$. The second disjunct thus holds by the second premise. It remains to show that
\begin{align*}
 \statepair{\state}{\heap},\beta,I \models \phi(\expr)
\end{align*}
This holds by the first premise.
Note that we have no identity of the called object, however the condition that these fields cannot be reassigned ensure that the binding of all objects in $\localtypesem$ is adhered to.
\item\textbf{\textsf{met}-\abs{get}} 
This case is analogous to the above. We stress that we do not rely on the read post-condition here, as we have one branch for the guarding formula and one for its negation.
\item \textbf{\textsf{met}-\abs{if}} 
This case is trivial.
\item\textbf{\textsf{met}-\abs{while}} 
We need to show that
\[\forall \trace \in \defeval{\xabs{while(e)}\{\xabs{s}\}\xabs{s'}}.~\trace,\beta,I \models \localtypesem(L^\ast.L') \]
First, we only consider the loop
\\\noindent\scalebox{0.9}{\begin{minipage}{\textwidth}
\begin{align*}
&\forall \trace \in \defeval{\xabs{while(e)}\{\xabs{s}\}}.~\trace,\beta,I \models \localtypesem(L^\ast) \iff\\
&\forall \trace \in \defeval{\xabs{while(e)}\{\xabs{s}\}}.~\trace,\beta,I \models \exists I \subseteq \mathbf{I}.~\exists a,b \in I.~\forall k\in \mathbf{I}.~(a \leq k \wedge k \leq b) \wedge\\
&\phantom{= \exists I \subseteq \mathbf{I}.~}\forall i_1,i_2 \in I.~\left((\forall l \in I.~ l \leq i_1 \wedge i_2 \leq l) \rightarrow \localtypesem'(\methtype)[n \in \mathbf{I} \setminus i_1 \leq n \wedge n \leq i_2)]\right)\\
\end{align*}
\end{minipage}}

We choose $I$ such that each $0\in I,|\trace|\in I$ and furthermore exactly those indices are included, which are generated by the start of another iteration.
Now we must show that for each consecutive pair $i_1,i_2 \in I$ the inner translation holds.

First we show that $J$\footnote{It is denoted $I$ in the rule, we use $J$ here because $I$ is already used as the index set.} is a loop invariant, i.e., that it holds at every $\trace[i]$ if $i \in I$, by induction on the subset of naturals in $I$.
For the induction base, the very first premise establishes that $J$ holds at $\trace[0]$.
For the infuction step, we may assume that if $J$ holds at $\trace[i]$ that it holds at $\trace(i')$ where $i'$ is the next element of $I$ after $i$. This is established by the second premise.
Note that here we use post-condition reasoning, not the method type! 

Finally, the third premisse establishes that if started in a state where $J$ holds, then \abs{s} realizes \methtype. As \abs{s} is executed within the loop and by the above argument, $J$ holds in all such states, this establishes
that the inner translation holds for each consecutive pair $i_1,i_2 \in I$.

Finally, we observe that the state where the last iteration\footnote{We remind that we only consider terminating runs} ends is in $I$, thus $J$ holds in this state. Thus, \abs{s'} starts in a state where $J$ holds and thus by the forth premise the continuation of the loop realizes the continuation of the type.

\item \textbf{Rewrite Rules} The first rule is trivial, for the others we observe that all types admit a prefix consisting only of \noev events (see first case) as well as a suffix with the same condition.
\end{itemize}
\end{proof}
\begin{corollary}
If (1) for every method \methodname  with type $?\methodname(\phi).\methtype_{\methodname}$ the formula $\localtypescheme(\methodname)$ is valid and (2) the obligation scheme is consistent, 
then for every selected trace $\trace$ of any method \methodname, the trace after the invocation reaction event follows its type: 
\[\trace[2..|\trace|],I,\emptyset \models \localtypesem(\methtype_{\methodname})\]
\begin{proof}
The proof obligation establishes the above property if for every method \methodname the formula $\phi_{\methodname}$ holds in the first state where it is scheduled. i.e., if every process of $\methodname$ agrees on a state $\defstatepair$
with 
\[\defstatepair,\emptyset,I \models \phi_{\methodname}\]
For this we observe that the first state is determined by the call parameters, the heap and the LA semantics of the method. 
In the type system presented here, $\phi_{\methodname}$ includes only the call parameters, not the heap, thus the precondition does not restrict the possible heaps.
Now, if the scheme is consistent, then every process is added after call such that the first state fulfills $\phi_\methodname$.
\end{proof}
\end{corollary}


%% file: newtotalrules.tex
\begin{figure}[bth]
\scalebox{0.9}{
\begin{minipage}{\textwidth}
\hspace{-4mm}
\scalebox{0.95}{
\begin{minipage}{0.5\textwidth}
\SINFER{\textsf{met}-V}{
\Gamma \fCenter \{U\}\{\xabs{v} := \xabs{e}\}[\xabs{s} \localsim \mathsf{L}], \Delta
    }{
\Gamma \fCenter \{U\}[\xabs{v = e; s} \localsim \mathsf{L}], \Delta
}    
\end{minipage}
}
\hspace{-4mm}
\scalebox{0.95}{
\begin{minipage}{0.5\textwidth}
\SINFER{\textsf{met}-F}{
\Gamma \fCenter \{U\}\{\xabs{heap} := \mathsf{store}(\xabs{heap},\xabs{f},\xabs{e})\}[\xabs{s} \localsim \mathsf{L}], \Delta
    }{
\Gamma \fCenter \{U\}[\xabs{this.f = e; s} \localsim \mathsf{L}], \Delta
}    
\end{minipage}
}

\vspace{2mm}
\scalebox{0.95}{\def\fCenter{\ }
\begin{minipage}{0.5\textwidth}
\SSINFER[$v$ fresh]{\textsf{met}-\abs{get}}{
\Gamma \Rightarrow \{U\}\{\xabs{v} := v\}(\phi(v) \rightarrow [\xabs{s} \localsim \mathsf{L_1}]), \Delta
}{
\Gamma \Rightarrow \{U\}\{\xabs{v} := v\}(\neg\phi(v) \rightarrow [\xabs{s} \localsim \mathsf{L_2}]), \Delta
\qquad
   \Rightarrow [ \xabs{v = e.get}_i\xabs{;s} \ptsim  \{\many{\methodname}\}]
}{
\Gamma \Rightarrow \{U\}[\xabs{v = e.get}_i\xabs{;s}\fCenter \localsim \&(\many{\methodname},\phi)\{\mathsf{L}_1,\mathsf{L}_2\}], \Delta
}    
\end{minipage}
}

\SINFER{\textsf{met}-\abs{while}}{
\Gamma \fCenter \{U\}I, \Delta
\qquad 
I,\xabs{e} \Rightarrow [\xabs{s} \postsim I ]
\qquad 
I,\xabs{e} \Rightarrow [\xabs{s} \localsim \textsf{L} ]
    \qquad
    I,\neg\xabs{e} \Rightarrow [\xabs{s'} \localsim \mathsf{L}'], \Delta
    }{
\Gamma \fCenter \{U\}[\xabs{while e do s od s'} \localsim \mathsf{L}^\ast.\mathsf{L}'], \Delta
}    

\SSINFER[$I_1 \cup I_2 \subseteq I$]{\textsf{met}-\abs{if}}{
\Gamma \fCenter \{U\}(\xabs{e} \rightarrow [\xabs{s;s''} \localsim \oplus\{\mathsf{L}_1\}_{i \in I_1}]), \Delta
    }{
\Gamma \fCenter \{U\}(\neg\xabs{e} \rightarrow [\xabs{s';s''} \localsim \oplus\{\mathsf{L}_2\}_{i\in I_2}]), \Delta
    }{
\Gamma \fCenter \{U\}[\xabs{if e then s else s' fi s''} \localsim \oplus\{\mathsf{L}_i\}_{i \in I}], \Delta
}

\SINFER[$f$ fresh]{\textsf{met}-call}{
    \Gamma \Rightarrow \{U\}\phi(\xabs{e}), \Delta
    \qquad
\Gamma \fCenter \{U\}\{\xabs{v} := f\}[\xabs{s} \localsim \mathsf{L}], \Delta
    }{
\Gamma \fCenter \{U\}[\xabs{v = O!m(e); s} \localsim  \xabs{O!m}(\phi).\mathsf{L}], \Delta
}    

\scalebox{0.95}{
\begin{minipage}{0.5\textwidth}
\SINFER{\textsf{met}-\abs{return}}{
\Gamma \fCenter \{U\}\{\xabs{result} := \xabs{e}\}\phi, \Delta
    }{
\Gamma \fCenter \{U\}[\xabs{return e} \localsim \downarrow\!(\phi)], \Delta
}    
\end{minipage}
}
\hspace{10mm}
\scalebox{0.95}{
\begin{minipage}{0.5\textwidth}
\SINFER{\textsf{met}-\abs{skip}}{
    }{
\Gamma \fCenter \{U\}[\xabs{skip} \localsim \mathsf{skip}], \Delta
}    
\end{minipage}
}
\end{minipage}
}

\[\methtype \leftrightsquigarrow \localactchoice{\methtype} \qquad \mathsf{skip}.\methtype \leftrightsquigarrow \methtype\qquad \methtype.\mathsf{skip} \leftrightsquigarrow \methtype\]
\caption{Rules for \localtype. We remind that the sets $I_1,I_2$ are defined as non-empty. For simplicity, we assume that every branch and every loop body implicitly ends in \abs{skip}.}
\label{fig:totalrules}
\end{figure}